\documentclass[journal]{IEEEtran}
\ifCLASSINFOpdf
\else
\fi
\usepackage{verbatim}
\usepackage{amsmath,amsthm,bbm}
\usepackage{amsfonts,amssymb}
\usepackage[noend]{algpseudocode}
\usepackage{algorithmicx,algorithm}
\usepackage{multirow}
\usepackage{graphicx}
\usepackage{epstopdf}
\usepackage{subfigure}
\usepackage{parskip}
\usepackage{bm}
\usepackage{color}
\usepackage{hyperref}
\usepackage{etoolbox}
\usepackage{dsfont}
\makeatletter
\patchcmd{\@makecaption}
{\scshape}
{}
{}
{}
\makeatother

\newtheorem{thm}{Theorem}%[section]
\newtheorem{cor}{Corollary}%[section]
\newtheorem{lem}{Lemma}%[section]
\newtheorem{exam}{Example}%[section]
\newcommand{\lemref}[1]{Lemma~\ref{#1}}

\newcommand{\setsep}{\;:\;}

\def\sph{\mathbb{S}^{2}}
\def\RR{\mathbb{R}}

\def\CB{\mathcal{B}}

\def\CH{\mathcal H}
\def\CF{\mathcal F}
\def\CV{\mathcal V}
\def\CW{\mathcal W}

\def\s{\sigma}
\def\NN{\mathbb{N}}
\def\f{\frac}

\def\b{\bm}

\def\Span{\operatorname{span}}

\begin{document}
\title{\huge Convolutional Neural Networks for Spherical Signal Processing via Area-Regular Spherical Haar Tight Framelets}
\author{Jianfei Li, Han Feng$^{\ast}$, and Xiaosheng Zhuang
\thanks{$^\ast$ Corresponding author.}
\thanks{J. Li, H. Feng, and X. Zhuang are with the Department of Mathematics, City University of Hong Kong, Tat Chee Avenue, Kowloon Tong, Hong Kong. Email: {jianfeili2-c@my.cityu.edu.hk, hanfeng@cityu.edu.hk, xzhuang7@cityu.edu.hk}.}
}
% stops a space


\markboth{For Special Issue on Deep Neural Networks for Graphs: Theory, Models, Algorithms and Applications}%
{Shell \MakeLowercase{\textit{et al.}}: Convolutional Neural Networks for Spherical Signal Processing via Spherical Haar Tight Framelets}



\maketitle

\begin{abstract}
%Wavelets/framelets have been well established for the Euclidean	spaces and widely applied in the processing and analysis of images while the extension to other kind of domains, like spheres, is in demand and possess the similar benefits.
In this paper, we develop a general theoretical framework for constructing Haar-type tight framelets on any compact set with a hierarchical partition. In particular, we construct a novel area-regular hierarchical partition on the 2-sphere and  establish its corresponding spherical Haar tight framelets with directionality. We conclude by evaluating and illustrate the effectiveness of our area-regular spherical Haar tight framelets in several denoising experiments. Furthermore, we propose a convolutional neural network (CNN) model for spherical signal denoising which employs the fast framelet decomposition and reconstruction algorithms. Experiment results show that our proposed CNN model outperforms threshold methods, and processes strong generalization and robustness properties.
\end{abstract}


% \fi
\IEEEpeerreviewmaketitle

\section{Introduction}
%\label{sec:intr}
\IEEEPARstart{W}{avelet}/framelet analysis, see e.g., \cite{book:chui,book:daubechies, book:han, book:mallat}, has been one of the central topics in applied and computational harmonic analysis and has achieved remarkable success in many real-world applications such as signal/image processing, computer graphics, numerical solutions of PDEs, and so on. Typical wavelet/framelet systems for image processing are developed on  Euclidean spaces $\RR^d$, where the signals are often sampled regularly, e.g., equally spaced samples. However, in practice, signals defined on a spherical surface rather than in Euclidean spaces arise in various situations like astrophysics \cite{astro1, astro2}, computer vision \cite{compu_vision_1} and medical imaging \cite{medical1,medical2}. Extending the theory and algorithms to spherical data is in great demand.
In the past decades, spherical framelets and filter banks are constructed and investigated based on spherical harmonics, rotation groups and Euler angles. For example, in \cite{Yeo} the authors developed theoretical conditions for the invertibility of filter banks under spherical harmonics and continuous spherical convolution while in \cite{WZ:ACHA}, the authors constructed  semi-discrete spherical tight framelets based on spherical harmonics and quadrature rules. More examples can be found in \cite{sph_w_1, sph_w_2,  sph_w_3, sph_w_4, sph_w_5, sph_w_6}.

In this paper, we aim at the construction and applications of simple but efficient Haar-type systems for data defined on a non-Euclidean domain and, specially, on the 2-sphere $\sph$. Our work is inspired by \cite{xiao2020adaptive}, in which the authors extended the most simple yet elegant Haar orthonormal wavelet system \cite{Haar1910} to the construction of Haar tight framelets on any compact set $\Omega\subseteq \RR^d$,  from the point of view of the underlying {hierarchical  partition}. Such a construction can be built easily with flexibility and it brings many advantages especially for directional and (di)graph representations \cite{HLZ:AML,Li:DHF,xiao2020adaptive}. By further exploiting the hierarchical partition, in this paper, we present a general framework for the construction of Haar-type tight framelet systems on $\Omega$. In particular,
on $\sph$, we  introduce  a novel hierarchical partition of our own through utilizing a bijective mapping that maps a unit square to part of the sphere, see Figure~\ref{fig:T}. This particular partition enjoys several desirable properties:
\begin{itemize}
	\item[{1)}] Area regular. The  partition  is  area-regular in the sense that after each partition children blocks coming from the same parent always have the same area.
	
	\item[{2)}] Local directionality. Each block is refined to 4 sub-blocks, which enables the construction of framelet functions having one of the following local directions: horizontal, vertical, diagonal, and anti-diagonal.
	
	\item[{3)}] Simple signal representation and fast processing. As a result, we develop framelets on $\sph$ and  efficient framelet transform algorithms. Spherical signals can be easily represented by our framelet system for processing.
\end{itemize}

To evaluate and illustrate the effectiveness of our framelets, several numerical experiments are conducted for spherical signal denoising by thresholding and deep neural networks. In fact, the connection of deep learning architectures (especially CNNs) and wavelets have been extensively studied and achieved state-of-the-art performance in various tasks, see \cite{Ref39,CNN_Liu,  mallat:CNN1, mallat:CNN2,Ref38, CNN_Zhang} for instances. In particular, inspired by the architecture of U-Net \cite{U-net}, we propose a deep convolutional neural network (CNN) model, which fuses  the fast decomposition and reconstruction algorithms of our framelets. Compared to the  U-Net model, the architecture of our network (see Figure~\ref{nn}) comprises four ConvConvT cells in which 2D convolution, transpose convolution are applied, and possesses  three different favours. (i) We add our area-regular spherical Haar tight framelets which have directionality and guarantee perfect reconstruction. With these properties, we can get different information with different directions. (ii) We take addition and ReLU activation in the second and third ConvConvT cell to push the neural network to learn a thresholding-like denoising behavior. (iii) First and last ConvConvT cells are used for feature extraction and reconstruction, which share the same weights that significantly reduce computation cost. Moreover, such a structure has two superiorities: 1) relatively small number (20 thousands) of trainable parameters. In comparison, \cite{CNN_Liu} presents a multi-level wavelet-CNN (MWC-NN) architecture for classical image restoration and denoising with over 50 million trainable parameters. 2) independent of the sizes of input. For the training and testing of our model, we produce spherical datasets from  the MNIST, CIFAR10, and Caltech101. Experiment results show that our proposed CNN model outperforms threshold methods, and processes strong generalization and robustness properties.

%In this paper, the CNN model we proposed is independent of input size

The contribution of the paper is threefold. First, we provide a general framework for the construction of Haar-type tight framelets on any
compact set with a hierarchical partition. Second, we present a novel area-regular hierarchical partition on the sphere that leads to the directional Haar tight framelet systems having many nice properties. Last but not least, we demonstrate that  our Haar tight framelets can be used for signal processing and CNN-like models on the sphere.

The structure of the paper is as follows. In Section~\ref{sec:constr}, we present our main theoretical results on the construction of Haar-type tight frames on any compact set.  In Section~\ref{sec:sphere}, we introduce the specific design of hierarchical partition on the sphere that leads to our area-regular spherical Haar tight framelets. In Section~\ref{sec:exp}, we provide numerical examples including CNN applications for signal denoising on the sphere. Conclusion and further remarks are given in Section~\ref{sec:remarks}. Proofs are postponed to Section~\ref{sec:appendix}.

\section{Haar tight framelets on any compact set}
\label{sec:constr}

In this section, we lay out our main results that give a general framework for the construction of Haar tight framelets on any compact set $\Omega\subseteq\RR^d$ with a hierarchical partition.

%\IEEEpubidadjcol

\subsection{Construction of Haar tight framelets}
Let $\CH$ be a separable Hilbert space  with its inner product $\langle\cdot,\cdot\rangle$ and norm $\|\cdot\|$. We call a countable collection $\{e_k\}_{k\in\Lambda}\subset\CH$  a \emph{frame} if there exist  $0<c_1\le c_2<\infty$ such that
\begin{align*}
	c_1 \|f\|^2 \leq \sum_{k\in\Lambda} | \langle f, e_k \rangle |^2 \leq c_2 \|f\|^2\quad \forall f\in \CH.
\end{align*}
If $c_1 = c_2 := c$, then $\{e_k\}_{k\in\Lambda}$ is called \emph{tight} with frame bound $c$. It is well-known that
\begin{align*}%\label{tight_frame}
	f = \f{1}{c} \sum_{k\in\Lambda}  \langle f, e_k \rangle e_k \quad \forall f\in\CH,
\end{align*}
if and only if $\{e_k\}_{k\in\Lambda}$ is tight with frame bound $c$. Furthermore, if $\{e_k\}_{k\in\Lambda}$ is tight and $\|e_k\|=1$ for all $k\in\Lambda$, then one can show that $\{e_k\}_{k\in\Lambda}$ is an orthonormal basis for $\CH$.
For a positive integer $N\in\NN$, we denote $[N]:=\{1,\dots,N\}$.

The following lemma characterize a tight frame from an orthonormal system, which is the key to our main result Theorem~\ref{Thm2.2}. We leave its proof  to Section~\ref{sec:appendix}.
\begin{lem}\label{lem2.1}
	For $n,\ell\in\NN$, let $\CH$ be a separable Hilbert space and	
	$\{\phi_1,\phi_2,\ldots, \phi_\ell\}$ be an orthonormal system in $\CH$.
	Define the systerm $\Psi := \{ \psi_1, \psi_2, \dots, \psi_n \}$ by
	\begin{align*}
		\begin{bmatrix}
			\psi_1 \\ \vdots \\ \psi_n
		\end{bmatrix}
		:=
		\b{A}
		\begin{bmatrix}
			\phi_1 \\ \vdots \\ \phi_\ell
		\end{bmatrix},
	\end{align*}
	for some matrix $\b{A} \in \RR^{n\times \ell}$. Then %where $\b{A} \in \CC^{n\times \ell}$ is a full rank matrix with $n\geq l-1$ and $\phi_m := \phi_{B_m}$ for all $m$.
	\begin{enumerate}
		\item[\rm(i)] $\Psi$ is a tight frame  for $\mathcal{W} := \Span\Psi $ with frame bound $c$ if and only if $ \b{A} = \f{1}{c} \b{A} \b{A}^{\top} \b{A} $;% or equivalently, $ \b{B}: = \f{1}{c} \b{A}^*\b{A} $ is an orthogonal projection matrix;}
		
		\item[\rm(ii)] for some $\phi_0:=\sum_{i=1}^\ell p_i \phi_i$ %$\b{p}[\phi_1,\ldots, \phi_\ell]^T$
		with $\b{p}:=(p_1,\ldots,p_\ell)\in \RR^\ell$,  $$ \Span\{ \phi_1, \dots, \phi_\ell \} =  \Span \{ \phi_0 \} \oplus \Span\{\psi_1, \dots, \psi_n\} $$
		if and only if  $\b{A}\b{p} = \b{0}$ and there exists a matrix $  \b{Q} \in \RR^{\ell \times (n+1)}$ such that $\b{Q}\left[
		\begin{array}{c}
			\b{p}^\top \\
			\b{A} \\
		\end{array}
		\right]
		= \b{I}_\ell$,
		where $\b{I}_\ell$ is the $\ell\times \ell$ identity matrix.
	\end{enumerate}
\end{lem}

Next, we introduce our main result for the general construction of Haar tight framelets on a compact set. Let  $\Omega \subseteq \RR^d$ be a compact set associated with the Lebesgue measure $|E|$ for $E\subseteq \Omega$ a Lebesgue measurable set. We consider the Hilbert space $L_2(\Omega):= \{f:\Omega \rightarrow \RR: \|f\|_2:=\int_{\Omega} |f|^2 dx < \infty \}$ with inner product $ \langle f, g \rangle :=  \int_{\Omega} f(x)g(x) dx $ for $f,g\in L_2(\Omega)$. We first introduce the concept of hierarchical partitions given in \cite{xiao2020adaptive}.
For a compact subset $\Omega\subseteq \RR^d$, let $\{\mathcal{B}_j\}_{j\in\NN_0}$ with $\NN_0:=\NN\cup\{0\}$
be a family of subsets of $2^\Omega$. We call  $\{\mathcal{B}_j\}_{j\in \NN_0}$
a \emph{hierarchical partition} of $\Omega$ if the following three conditions are satisfied:
\begin{enumerate}
	\item[{a)}] Root property: $\mathcal{B}_0 = \{\Omega\}$ and each $\mathcal{B}_j$
	%$=\{R^j_\ell: \ell=1,\ldots, n_j\}$
	is a partition of $\Omega$  having  finite many  measurable sets with positive measures.
	\item[{b)}] Nested property: for any $j\in \NN$ and any (child) set $R_1 \in \mathcal{B}_j$, there exists a (parent) set $R_0 \in \mathcal{B}_{j-1}$ such that $R_1 \subseteq R_0$. In other word, partition $\mathcal{B}_j$ is a refinement of partition $\mathcal{B}_{j-1}$.
	\item[{c)}] Density property: the maximal number among diameters of sets in $\mathcal{B}_j$ tends to zero as $j$ tends to infinity.
\end{enumerate}

%
%Given a hierarchical partition $\{\CB_j\}_{j=0}^\infty$ of a compact subset $\Omega\subset \RR^d$, to be simple, for all $j\geq 1$ we assume every element in $\CB_{j-1}$ contains the same number of subblocks in $\CB_{j}$, denoted by $n_j\geq 1$. % by the cardinality of $\CB_j$  and $\CB_j=\{R_{\vec v}\}_{\vec v\in [n_1]\times [n_j]}$ for all $j$ such that
 Without loss of generality, we  assume that for all $j\geq 1$ every element in $\CB_{j-1}$ contains {the same number of sub-blocks in $\CB_{j}$}, denoted by $\ell_j\geq 1$. We remark that this is only for the purpose of convenience of notation and our results can be easily adapted to the case of blocks with different number of chilren sub-blocks.
We denote  $\Lambda_j:= [\ell_1]\times\cdots\times[\ell_j]$ to be an index set for labeling sets in $\CB_j$. Then we can write
\[
\CB_j=\{R_{\vec v}\subseteq \Omega\setsep {\vec v\in \Lambda_j}\}.
\]
By the nested property,  we have
$ R_{(\vec v,i)}\subseteq R_{\vec v}$ for $\vec v\in\Lambda_{j-1}$ and $i\in[\ell_j]$.  Now, for each  $\vec v\in \Lambda_j$, we can define a Haar-type scaling function
\begin{equation}
\label{eq:phi}
	\phi_{\vec v}:=\frac{\chi_{R_{\vec v}}}{\sqrt{|R_{\vec v}|}},
\end{equation}
and for some integer $n_j\geq 1$, we can define  $n_j$ Haar-type framelet functions
\begin{equation}\label{eq:Psi}
	\psi_{(\vec v,k)}=\sum_{i\in[\ell_j]}{a^{(\vec v)}_{k,i}}\phi_{(\vec v,i)},\quad k=1,\ldots, n_j,
\end{equation}
where $a^{(\vec v)}_{k,i}$ is the $(k,i)$-entry  of some matrix $\b A_{\vec v}\in\RR^{n_j\times \ell_j}$.
Fixed $L\in\NN_0$ and $c>0$, we can define a function system $\CF_L(\{\mathcal{B}_j\}_{j \in \mathbb{N}_0})$ as follow:
\begin{equation}
	\label{def:FS}
	\CF_L(\{\mathcal{B}_j\}_{j \in \mathbb{N}_0}) := \{\sqrt{c}\phi_{\vec u}\}_{\vec u\in \Lambda_L} \cup \{\psi_{(\vec v,k)}, k\in[n_j]\}_{j \geq L, \vec v \in \Lambda_j}.
\end{equation}
We call $\CF_L(\{\mathcal{B}_j\}_{j \in \mathbb{N}_0})$ a \emph{Haar framelet system} for $L_2(\Omega)$. If it is a tight frame for $L_2(\Omega)$, then we call it $\emph{Haar tight framelets}$ for $L_2(\Omega)$.

Applying Lemma~\ref{lem2.1} to the function system $\CF_L(\{\mathcal{B}_j\}_{j \in \mathbb{N}_0})$, we have the following theorem (see Section~\ref{sec:appendix} for its proof).

\begin{thm}\label{Thm2.2}
	Let $\{\CB_j\}_{j\in\NN_0}$ be a hierarchical partition for a compact set $\Omega\subseteq\RR^d$. Define the functions $\phi_{\vec v}$ and $\psi_{(\vec v,i)}$  with $\b A_{\vec v}:=(a^{(\vec v)}_{k,t})\in \RR^{n_j\times \ell_j}$ and $\b p_{\vec v}:=\frac{1}{ \sqrt{|R_{\vec v}|}}(\sqrt{|R_{(\vec v, 1)}|} ,\ldots, \sqrt{|R_{(\vec v, \ell_j)}|})\in\RR^{\ell_j}$ as in \eqref{eq:phi} and \eqref{eq:Psi}. If for all $j\geq 1$ and all $\vec v\in \Lambda_j$, the matrix $\b A_{\vec v}$ satisfies
	\begin{equation}\label{cond1}
		\b A_{\vec v}=\frac 1 c \b A_{\vec v}\b A_{\vec v}^\top\b A_{\vec v}
	\end{equation}
	and there exists
	$\b Q_{\vec v}\in \RR^{\ell_j\times (n_j+1)}$ such that
	\begin{equation}\label{cond2}
		\b Q_{\vec v}\binom{\b p_{\vec v}^\top }{ \b A_{\vec v}}=\b I_{\ell_j},
	\end{equation}
	then the system $\CF_L(\{\mathcal{B}_j\}_{j \in \mathbb{N}_0})$ as defined in \eqref{def:FS}  is a tight frame for $L_2(\Omega)$ with frame bound $c$ for all $L\in \NN_0$.
\end{thm}

Obviously, Theorem~\ref{Thm2.2} recovers results in \cite[Theorem~1]{xiao2020adaptive} since its matrix $\b A$ is a special case of \eqref{cond1}. Hence, our result can be applied for digraph signal representation, which is not the focus of this paper and we refer to \cite{xiao2020adaptive} for more details.

The following result can be easily derived from Theorem~\ref{Thm2.2} and it shows that  $\CF_L(\{\mathcal{B}_j\}_{j \in \mathbb{N}_0})$ can recover the case of orthonormal Haar bases for $L_2(\Omega)$. See, e.g., Example~\ref{example1}.

\begin{cor}\label{cor:tight-orth}
Retain notations in Theorem~\ref{Thm2.2} and assume that conditions \eqref{cond1} and \eqref{cond2} are satisfied with $c=1$. If in addition	the main diagonal entries of each $\b{A_{\vec v}A_{\vec v}}^\top$ are all equal to $1$, then the framelet system $\CF_L(\{\mathcal{B}_j\}_{j \in \mathbb{N}_0})$ is an orthonormal basis for $L_2(\Omega)$.
\end{cor}

%\begin{proof}
%Note that  the framelet system $\CF_L(\{\mathcal{B}_j\}_{j \in \mathbb{N}_0})$ is tight with frame bound $c=1$. Now the conclusion follows from the facts that $\{\phi_{\vec u}\}_{\vec u \in\Lambda_L}$ is an orthonormal system and the norm of each $\psi_{(\vec v,k)}$ is equal to 1.
%\end{proof}

%%%%%%%%%%%%%%%%%%%%%%%%%%%%%%%%%%%%%%%%%%

In the following, we illustrate some examples for matrices $\b A$ and $\b p$ which satisfy conditions \eqref{cond1} and \eqref{cond2} in Theorem~\ref{Thm2.2}. %The  domain we use is the unit square $\Omega=[0,1]^2$. We remark that one can use any other domain satisfying Theorem~\ref{Thm2.2}.
%In Examples~\ref{example1}--\ref{example3}.

%%%%%%%%%%%%%%%%%%%%%%%%%%%%%%%%%%%%%
%%%%%%%%%%%%%%%%%%%%%%%%%%%%%%%%%%%%%
\begin{exam}\label{example1}
	{\rm Consider $\Omega=[0,1]^2$ and in each step $\CB_{j-1}\rightarrow\CB_j$ of the hierarchical partition $\{\CB_j\}_{j\in\NN_0}$, the block $R_{\vec v}\in \CB_{j-1}$ is refined to 4 sub-blocks with the same area. That is, $\ell_j\equiv 4$ and $\Lambda_j=[4]^j$ for all $j\in\NN_0$.
		For all $\vec v\in [4]^j$, set
		$
		\b A_{\vec v}\equiv\b{A}:=
		\f{1}{2}\begin{bmatrix}
			1 & 1 &-1&-1 \\
			1&-1&1&-1 \\
			1&-1&-1&1
		\end{bmatrix}
		$ and $\b p_{\vec v}\equiv \b p:=\f{1}{2}[1,1,1,1]^\top$. Then, $\b A$ satisfies \eqref{cond1} with $c=1$ and \eqref{cond2} with $\b Q_{\vec v}\equiv \b Q := [\b p, \b A^\top]$. In this case, $n_j\equiv 3$. We have $3$ framelets $\psi_{(\vec v,k)}$ attached with each block $R_{\vec v}$ and its refinement $R_{(\vec v,k)}$ for $k=1,2,3$. The corresponding framelet system  $\CF_L(\{\mathcal{B}_j\}_{j \in \mathbb{N}_0})$ is indeed an orthonormal basis for $L_2([0,1]^2)$ in view of Corollary~\ref{cor:tight-orth}. In particular
if the partition of $[0,1]^2$ is dyadic, then   $\CF_L(\{\mathcal{B}_j\}_{j \in \mathbb{N}_0})$   coincide with the classical tensor product Haar orthonormal wavelets on $[0,1]^2$. See Figure~\ref{fig:frame1}.
		\begin{figure}[H]
			\centering
			\subfigure[dyadic partition]{\includegraphics[width=4cm]{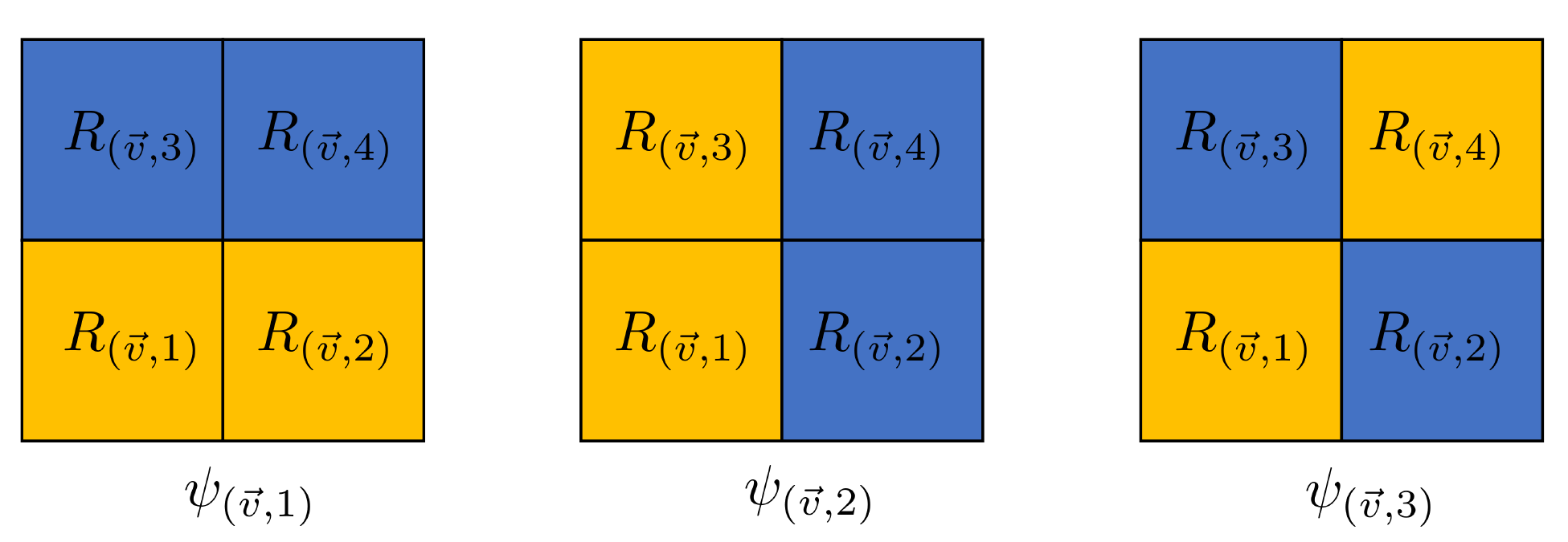}}
		\quad \subfigure[Irregular partition]{\includegraphics[width=4cm]{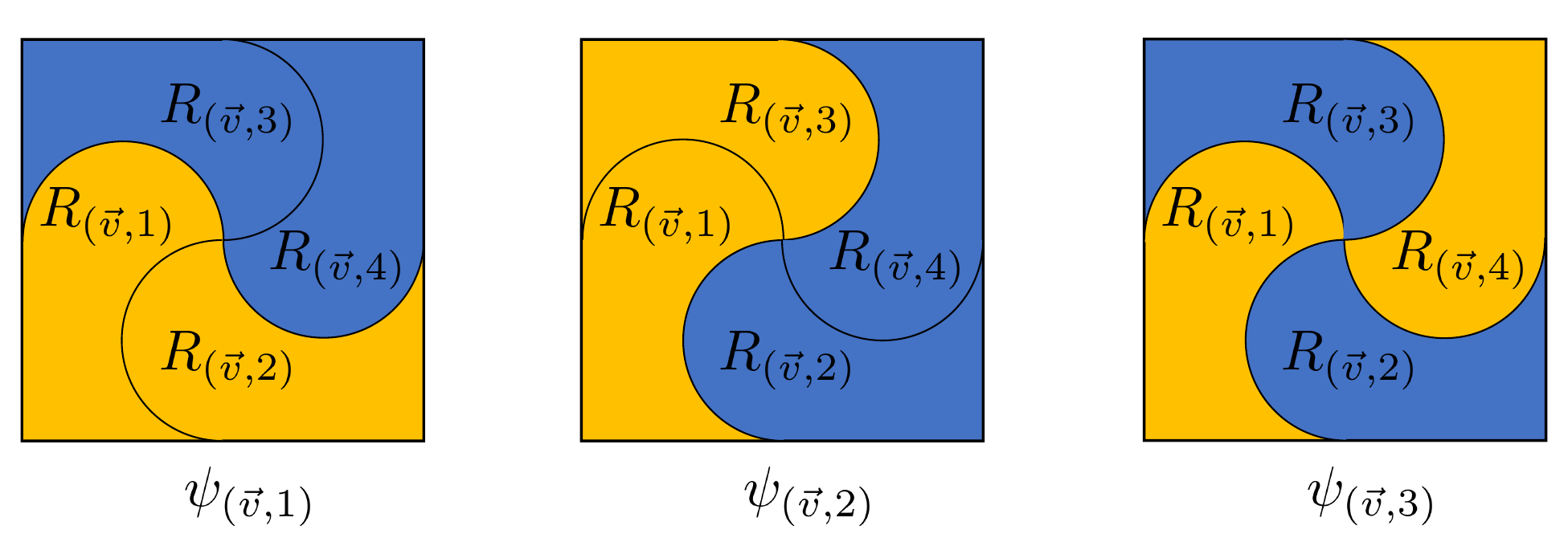}}
			\caption{Haar tight framelets in Example \ref{example1}. $\psi_{(\vec v,k)}$, $k=1,2,3$ generate the  Haar orthonormal wavelets on $[0,1]^2$. Each color represents the  value of the function. Left:  Dyadic partition. Right: Irregular partition.}
			\label{fig:frame1}
		\end{figure}
	}
\end{exam}

%%%%%%%%%%%%%%%%%%%%%%%%%%%%%%%%%%%%%
%%%%%%%%%%%%%%%%%%%%%%%%%%%%%%%%%%%%%
We remark that the matrix $\b A_{\vec v}$ in Theorem~\ref{Thm2.2} is more or less independent of the domain $\Omega$ and its hierarchical partition in the sense that we can use such an $\b A$ for the construction of Haar tight framelets in other domains. See below Example~\ref{ex:1D}.
\begin{exam}\label{ex:1D}
{\rm
Consider the same $\b A$ and $\b p$ in Example~\ref{example1} but $\Omega$ the unit interval. Again, the partition is of $4$ in each refinement. That is, $\Omega =[0,1]$ and $\Lambda_j=[4]^j$. Then  we have $3$ framelets $\psi^4_{(\vec v,k)}, k=1,2,3$ (see Figure \ref{fig:hw} right) attached with each sub-interval so that the corresponding framelet system  $\CF_L(\{\mathcal{B}_j\}_{j \in \mathbb{N}_0})$ is tight (with frame bound 1) for $L_2([0,1]^2)$. On the other hand, if $\Omega=[0,1]$ is with dyadic partition $\Lambda_j=[2]^j$, $\b A=\frac{1}{\sqrt{2}}[1,-1]$, and $\b p=\frac{1}{\sqrt{2}}[1,1]$, then we have the classical Haar orthonormal wavelets on $[0,1]$ (see Figure \ref{fig:hw} left).
\begin{figure}[H]
			\centering
			\includegraphics[width=7cm]{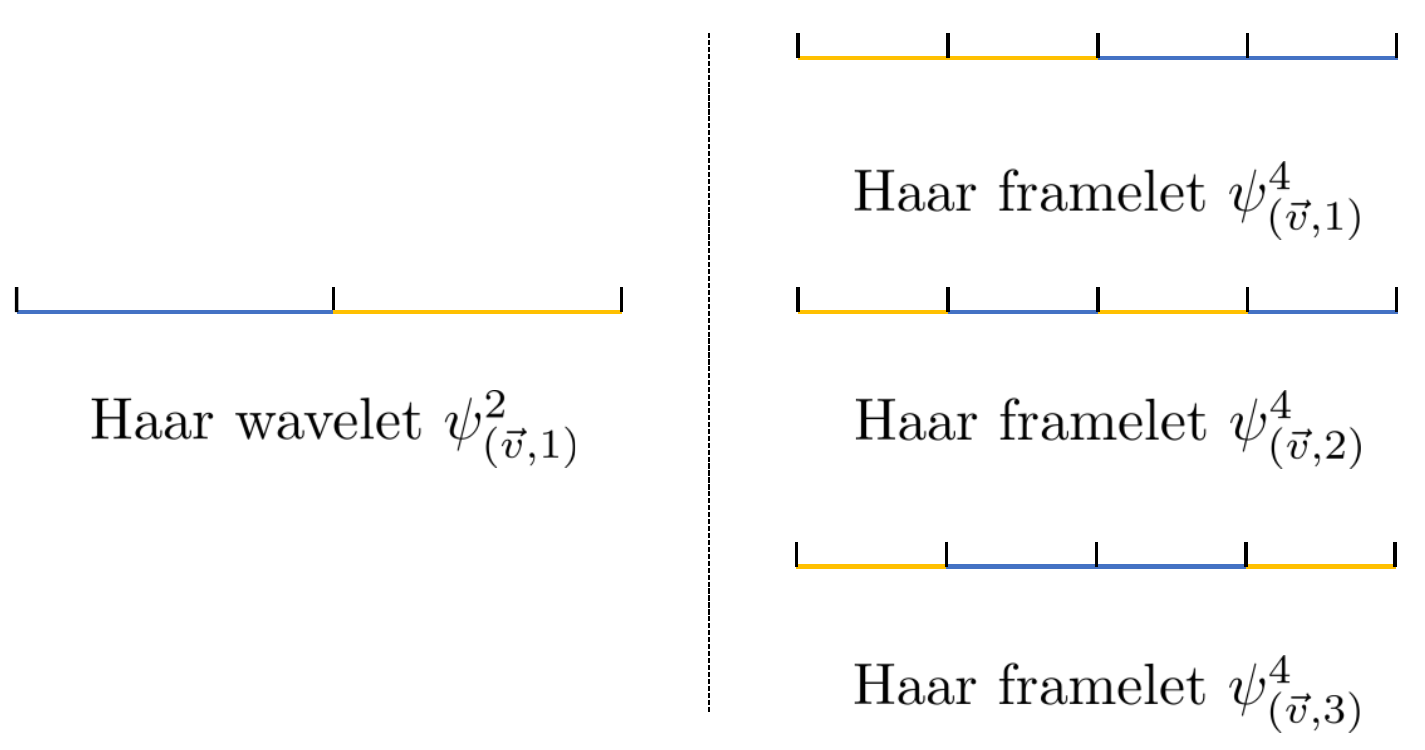}
			\caption{Left: Haar orthonormal wavelet on $[0,1]$ with dyadic partition. Right: Haar tight framelets on $[0,1]$ with $4$-refinement partition.}
			\label{fig:hw}
\end{figure}
}
\end{exam}

%%%%%%%%%%%%%%%%%%%%%%%%%%%%%%%%%%%%%
%%%%%%%%%%%%%%%%%%%%%%%%%%%%%%%%%%%%%
%\begin{exam}%\label{example2}
%	{\rm For an arbitrary compact set $\Omega$ with a hierarchical partition $\{\CB_j\}_{j\in\NN_0}$. Suppose the block $R_{\vec v}\in \CB_{j-1}$ is with $m$ sub-block $R_{\vec v, k}, k\in[m]$. Set the matrix $\b A_{\vec v} =(a_{i,k})_{i\in[n]; k\in [m]}$ with $n:=\binom m 2$ by
%\[
%a_{i,k}=a_{(i_1,i_2),k}=
%\begin{cases}
%\sqrt{\frac{|R_{(\vec v, i_2)}|}{|R_{\vec v}|}} & \mbox{if } k=i_1,\\
%-\sqrt{\frac{|R_{(\vec v, i_1)}|}{|R_{\vec v}|}} & \mbox{if } k=i_2,\\
%0 & \mbox{otherwise},
%\end{cases}
%\]
%where for each $i$, the pair $(i_1,i_2)$ is uniquely determined by the relation $1\le i_1<i_2\le m$ and $i=\frac{(2m-i_1)(i_1-1)}{2}+i_2-i_1$. Then $\b A_{\vec v}$ satisfies \eqref{cond1} with $c=1$. This recovers results in		 \cite[Theorem~1]{xiao2020adaptive}. For more details and its relation to digraph signal representation, we refer to \cite{xiao2020adaptive} and references therein.
%	}
%\end{exam}

%%%%%%%%%%%%%%%%%%%%%%%%%%%%%%%%%%%%%
%%%%%%%%%%%%%%%%%%%%%%%%%%%%%%%%%%%%%
We next gives an example which is not covered by results in \cite{xiao2020adaptive}.
\begin{exam}\label{example3}
	{\rm Consider a partition of an triangle having 3-refinement partition with same areas. That is $\ell_j\equiv 3$ and $\Lambda_j=[3]^j$ for all $j\in\NN_0$. For all $\vec v\in[3]^j$,
		set
		\begin{align*}
			\b{A}_{\vec v}\equiv \b A :=
			\f{1}{3}\begin{bmatrix}
				2&-1&-1\\
				-1&2&-1\\
				-1&-1&2
			\end{bmatrix},
		\end{align*}
		and $\b p_{\vec v}\equiv \b p:=\f{\sqrt{3}}{3}[1,1,1]^\top$. One can easily verify that $\b A$ defines a tight frame with frame bound $1$, which is different from the classical one and that in \cite{xiao2020adaptive}. Figure \ref{fig:frame2} visualizes this example.
		\begin{figure}[H]
			\centering
			\includegraphics[width=7cm]{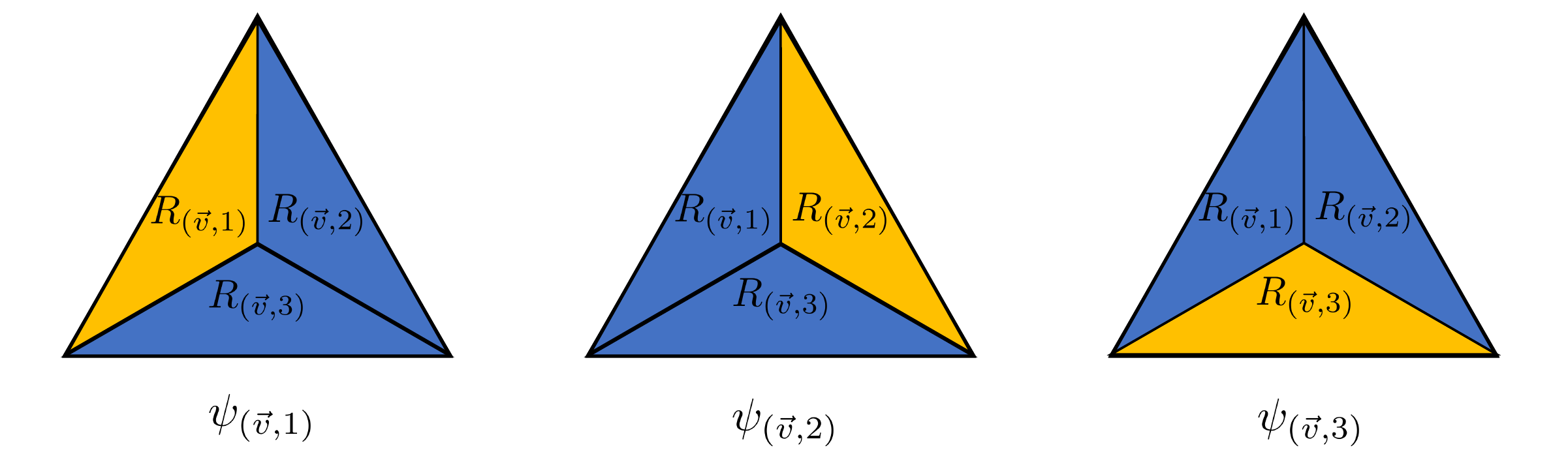}
			\caption{Haar tight framelets in Example \ref{example3} on a triangle area. For each block $R_{\vec v}$ there are 3 framelets $\psi_{(\vec v, k)}$ associated with the sub-blocks $R_{(\vec v, k)}$ for $k\in[3]$. }
			\label{fig:frame2}
		\end{figure}
	}		
\end{exam}

In fact, Example~\ref{example3} is a special case of the following result. Moreover, the matrix $\b A$ below is related to the area-regular Haar tight framelets we
consider in Section~\ref{sec:sphere}. See Section~\ref{sec:appendix} for its proof.
\begin{cor}
\label{cor:area-regular}
Given a vector $\b w \in \RR^\ell$ such that $\b w^\top \b 1 = \b 0$, suppose there are $n$ distinct permutations of $\b w$. Let $\b A \in \RR^{n\times \ell}$ be a matrix whose rows are the permutations of $\b w$. If the row rank of $\b A$ is equal to $\ell -1$, then for some constant $c\neq 0$, $\b A=\frac 1 c \b A\b A^\top\b A$, and there exists $\b Q\in \RR^{\ell \times (n+1)}$ such that $\b Q \binom{\b 1^\top}{\b A}=\b I_\ell$.
%$\b A$ satisfies condition~\eqref{cond1} in Theorem~\ref{Thm2.2} and $\b A \b A^\top$ has same main diagonal entries.
\end{cor}

%{\red Formulate a result of the following remark and continue to example 7.}
%\begin{rem}
%	In fact, the construction for matrices $\b A$ and $\b Q$ is flexible and simple, in which the key point is for $\b A$. One possible method can be
%	$\b{A}:= \f{1}{\sqrt{c}} \b{V} \b{U}^T $ with $\b{U} \in \RR^{\ell \times m}$ and $\b{V} \in \RR^{n \times m}$ being two matrices whose columns are orthonormal and $c\neq 0$, $n+1\geq \ell\geq m$. It is easy to verify that $ \b{A} = \f{1}{c} \b{AA}^T\b{A}$. Select $\b{p}\in row(\b A)^{\bot}$ such that $\b{p} \in (0,1)^\ell$ and
%	set $\b P=
%	\begin{bmatrix}
%		\b p^T \\
%		\b A \\
%	\end{bmatrix}
%	$. Then make $\b Q\in \RR^{\ell\times (n+1)}$ such that $\b Q\b P =\b I_{\ell}$. One possible choice for $\b Q$ could be $(\b P^T \b P)^{-1}\b P^T$ once $\b P$ has full rank columns.%, which can be implemented by using quasi-inverse, see for instance \cite{quasi-inverse}.
%\end{rem}

More generally, the matrices $\b A$ and $\b Q$ are quite flexible, in which the key construction is for $\b A$. The following result gives a general construction for $\b A$ and $\b p$ satisfying conditions in Theorem~\ref{Thm2.2} and can be verified directly.
\begin{cor}
	Assume that $n+1\geq \ell\geq m$ and $c\neq 0$. Given a unit vector $\b p\in\RR^\ell$ with all elements positive and matrices $\b{U} \in \RR^{\ell \times m}$, $\b{V} \in \RR^{n \times m}$ whose columns are orthonormal. Let $\b{A}:= \f{1}{\sqrt{c}} \b{V} \b{U}^\top $. If it satisfies $\b A \b p = \b 0$, then there exists $\b Q\in \RR^{\ell\times (n+1)}$ such that $\b Q\b P =\b I_{\ell}$, where $\b P=
	\begin{bmatrix}
		\b p^\top \\
		\b A \\
	\end{bmatrix}$. In particular, one can choose  $\b p \in col(\b U)^{\bot}$ and $\b Q=(\b P^\top \b P)^{-1}\b P^\top$.
\end{cor}
Applying the above result, another example for $\b A$ and $\b p$ can be constructed as follows.
\begin{exam}\label{example4}
{\rm	Let
	\begin{align*}
		\b{U} :=
		\begin{bmatrix}
			1/\sqrt{5}&0\\
			-\f{2\sqrt{5}}{5}&0\\
			0&\f{3}{5}\\
			0&-\f{4}{5}
		\end{bmatrix}, \,\,\,
		\b{V} :=
		\begin{bmatrix}
			\f{\sqrt{3}}{3}&\f{\sqrt{2}}{2}\\
			\f{\sqrt{3}}{3}&-\f{\sqrt{2}}{2}\\
			\f{\sqrt{3}}{3}&0
		\end{bmatrix},
	\end{align*}
	and define $\b A_{\vec v}\equiv\b A:= \b V \b U^\top$. It is easy to see that $\b A = \b A \b A^\top \b A$ and $\b p_{\vec v}\equiv\b p := \f{1}{\sqrt{30}}[2,1,4,3]^\top$ satisfies $\b A \b p = \b 0$. Such $\b A$ and $\b p$ can associate with a non area-equal hierarchical  partition  that produces Haar tight framelets. See an example of partition of $\Omega=[0,1]^2$  in Figure~\ref{fig:frame3}.
	\begin{figure}[H]
		\centering
		\includegraphics[width=6cm]{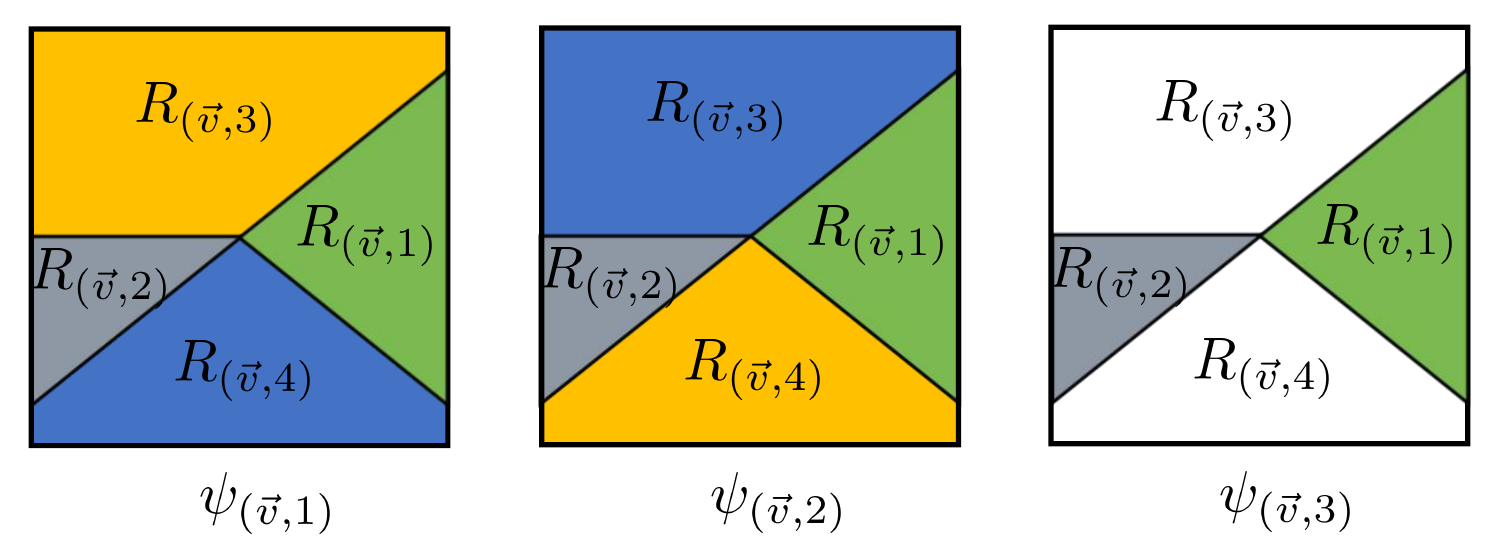}
		\caption{Non area-equal Haar tight framelets in Example \ref{example4} on $[0,1]^2$.}
		\label{fig:frame3}
	\end{figure}
}
\end{exam}

\subsection{Fast Haar framelet transforms}

In real-world problems, signals are discrete samples of certain underlying functions. For $j\in\NN_0$, let $\mathcal{V}_{j} := \Span\{ \phi_{\vec v} \setsep \vec v\in \Lambda_j \}$ and $\mathcal{W}_{j} := \Span\{ \psi_{(\vec v,k)} \setsep \vec v\in\Lambda_j, k\in [n_j] \}$. In view that $ \bigcup_{j = 0}^{\infty} \mathcal{V}_j $ is dense in $ L_2(\Omega) $,  we can use simple functions to represent (approximate) the ground truth signals for signal processing.  Hence, in this section we mainly consider signal decomposition and reconstruction by Haar tight framelets.

In what follows, we consider the signal $f$  in $\CV_J$ and its decomposition and reconstruction  in $\CV_L\oplus \CW_L\oplus \CW_{L+1}\oplus\cdots\oplus \CW_{J-1}$. We begin by considering signal decomposition for $J = L+1$ and iteratively decomposing signals we get all coefficients of tight Haar framlets for any $J$. Assuming that
$f\approx f_{\CV_J}=\sum_{\vec v\in \Lambda_J} c^{(J)}_{\vec v} \phi_{\vec v}\in \CV_J$ for some large $J>0$.
To be convenient, we assume $\b A:=\b A_{\vec v}$, $\b P:= \b P_{\vec v}=\binom{\b p^\top}{\b A_{\vec v}}$ and $\b Q:=\b Q_{\vec v}$ for all $\vec v\in \bigcup_{j=0}^J\Lambda_j$. Then
\[
\begin{small}
\begin{aligned}
	f_{\CV_J}&=\sum_{\vec v\in \Lambda_J} c^{(J)}_{\vec v} \phi_{\vec v}\\
	&=\sum_{\vec u\in \Lambda_{J-1}}\sum_{i\in [\ell]} c^{(J)}_{(\vec u,i)}  q_{i,1}\phi_{\vec u}+\sum_{\vec u\in \Lambda_{J-1}}\sum_{s=1}^n\sum_{i\in [\ell]} c^{(J)}_{(\vec u,i)}  q_{i,s+1}\psi_{(\vec u,s)}\\
	&=:f_{\CV_{J-1}}+f_{\CW_{J-1}},
\end{aligned}
\end{small}
\]
where $ q_{i,k}$ is the $(i,k)$ entry of $\b Q$. In the second step we used the definitions \eqref{eq:phi} and \eqref{eq:Psi} and \lemref{lem2.1}. Proceeding the step iteratively, we can establish the decomposition formula for any  $0\leq L<J$ (see Figure~\ref{partition}),
\begin{align*}
	f_{\CV_J}=&\sum_{\vec u\in \Lambda_L} c^{(L)}_u \phi_{\vec u}+\sum_{j=L}^J \sum_{\vec w\in \Lambda_j }\sum_{i\in[n]} d^{(j)}_{(\vec w,i)} \psi_{(\vec w,i)}\\
	=&f_{\CV_L}+f_{\CW_L}+\cdots+f_{\CW_{J-1}}.
\end{align*}

\begin{figure}[H]
\centering
\includegraphics[width=3in]{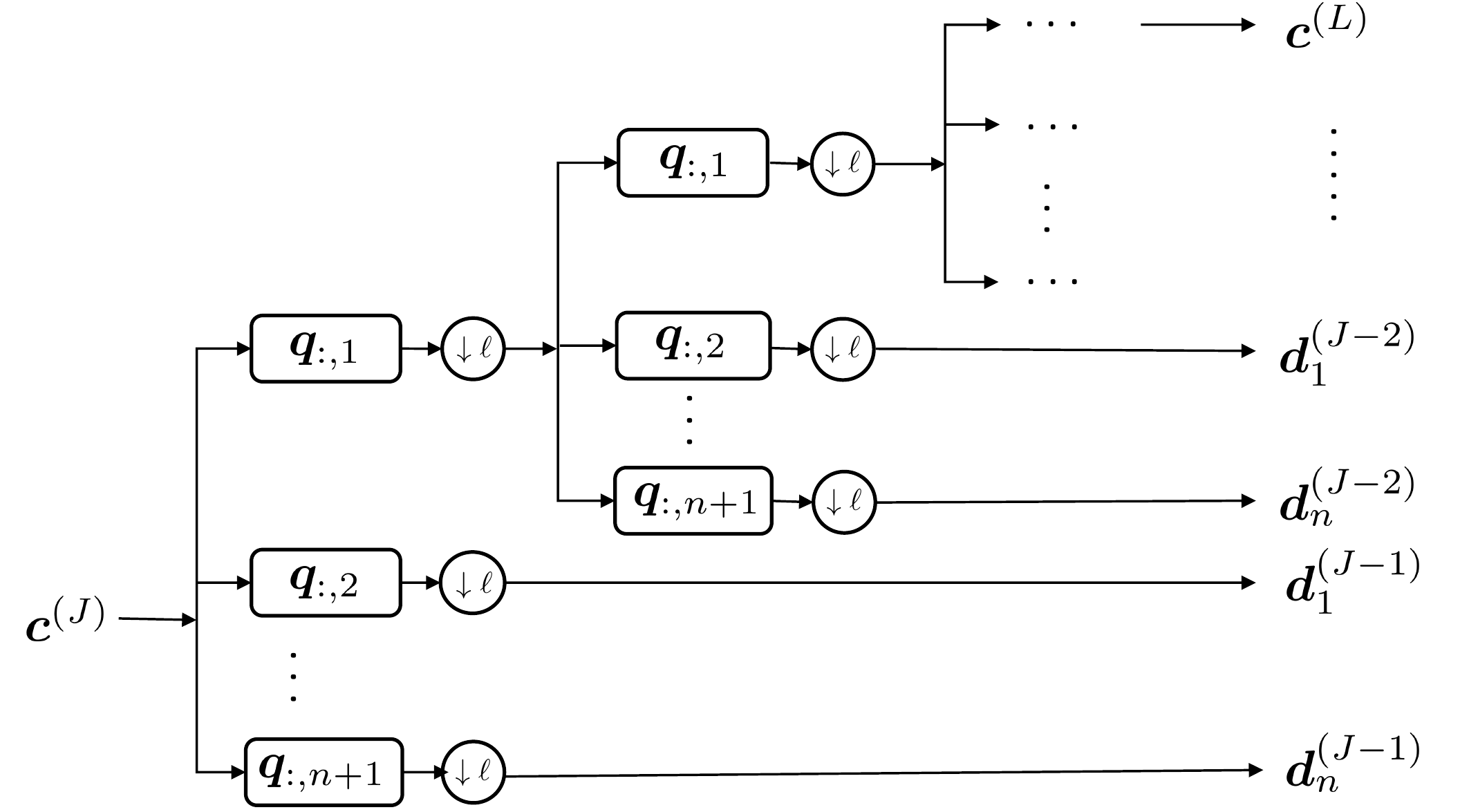}
\caption{Listing all coefficients as $\b c^{(J)}$ and constructing analysis filter bank $\{\b q_{:,1},\dots,\b q_{:,n+1}\}$ formed by columns of $\b Q$. Decomposition procedure is by applying the analysis filter bank of convolution, repeatedly. }
\label{partition}
\end{figure}
Conversely, given $f_{\CV_L}, f_{\CW_L}, \ldots, f_{\CW_{J-1}}$, we can reconstruct a function $f_{\CV_J}$ in $\CV_J$ by calculating the coefficient for every $\phi_{\vec v}$, $\vec v\in \Lambda_J$ such that
$f_{\CV_J}=f_{\CV_L}+f_{\CW_L}+\cdots+f_{\CW_{J-1}}$.
Assume that
$f_{\CV_L}=\sum_{\vec v\in \Lambda_L} c^{(L)}_{\vec v} \phi_{\vec v}$
and
$f_{\CW_L}=\sum_{\vec v\in \Lambda_L}\sum_{i=1}^n d^{(L)}_{(\vec v,i)} \psi_{(\vec v,i)}$. Let
\[
c^{(L+1)}_{(\vec v,j)}=c^{(L)}_{\vec v} p_{1,j}+\sum_{i=2}^{n+1} p_{i,j}d^{(L)}_{(\vec v,i-1)}.
\]
Then, we have
\[f_{\CV_{L+1}}=\sum_{\vec v\in\Lambda_L}\sum_{j\in[\ell]}c^{(L+1)}_{(\vec v,j)} \phi_{(\vec v,j)}. \]
Proceeding the step iteratively, one can get the reconstruction of coefficients for $f_{\CV_J}$ (see Figure~\ref{apartition}).

\begin{figure}[H]
\centering
\includegraphics[width=3in]{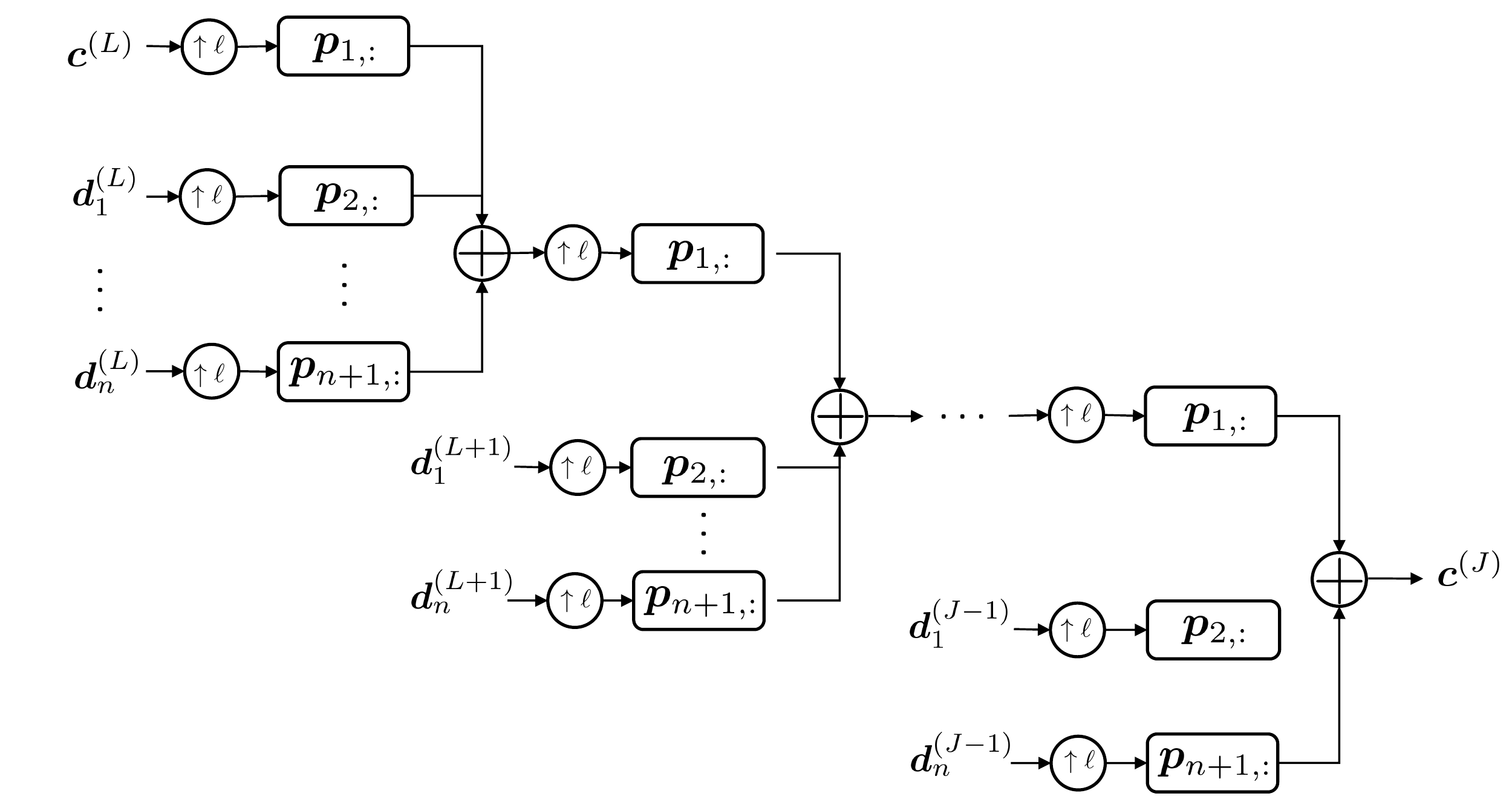}
\caption{Listing all coefficients as $\b c^{(J)}$ and $\b d^{(J)}_{i}$ for $i \in [n]$ and constructing synthesis filter bank $\{\b p_{1,:},\dots,\b p_{n+1,:}\}$ formed by rows of $\b P$. Reconstruction procedure is by applying the synthesis filter bank for convolution, repeatedly. }
	\label{apartition}
\end{figure}

We next give the computational complexity of our decomposition and reconstruction algorithms. Since input has $N := \ell^J$ elements for some integer $J$ in the decomposition procedure, one level/time decomposition requires $\ell^{J-1}(\ell-1)$ and $\ell^{J-1}$ operations for additions and multiplications respectively. If consider decomposing input from level $J$ to level $L$, i.e. $(J-L)$-level decomposition, one need $O(N)$ operations for additions and multiplications during decomposition. Similarly, considering reconstruction with $n+1$ inputs with size $\ell^L$ to level $J$, we need $O(N)$ operations.

\section{Area-Regular Spherical Haar tight framelets}
\label{sec:sphere}
In this section, based on the results in Section~\ref{sec:constr}, we realize a specific Haar tight framelets on the important domain $\Omega=\sph$, the 2-sphere. Thanks to Theorem~\ref{Thm2.2}, we only need to focus on the design of a hierarchical partition on $\sph$. We introduce here a novel construction of an area-regular hierarchical partition having many nice properties as mentioned in the introduction.

Given an integer $j$, we  next establish our novel area-regular hierarchical partition $\{\CB_j\}_{j\in\NN_0\cup\{-1\}}$ with $\CB_{-1}:=\{\sph\}$,
$\CB_0=\{R^{(s)}\}_{s\in[6]}$, and $\CB_j=\left\{R^{(s)}_{\vec v}\right\}_{\vec v\in [4]^j, s\in [6]}$ for all $j\in\NN$. The sets $R^{(s)}_{\vec v}$ satisfy the following properties:
\begin{enumerate}
	\item[(i)] $\bigcup_{s=1}^6 \bigcup_{\vec v\in [4]^k} R^{(s)}_{\vec v}=\sph$;
	\item[(ii)]  $\left|R_{\vec a} \cap R_{\vec b}\right|=0$ for all $\vec a, \vec b\in [4]^k$ and $\vec a\neq \vec b$;
	\item[(iii)] $R^{(s)}_{\vec a}\subset R^{(s)}_{\vec b}$ whenever $\dim \vec a=\dim \vec b+1$ and $\vec a=(\vec b, \ell)$ for some $\ell\in [4]$;
	\item[(iv)] $\left|R^{(s_1)}_{ \vec v_1}\right|=\left|R^{(s_2)}_{\vec v_2}\right|$
	once $\dim \vec v_1=\dim \vec v_2$, $s_1,s_2\in [6]$.
\end{enumerate}
Here $\dim \vec v$ is the length of the vector $\vec v$ and $[4]^0:=\{\emptyset\}$. Such a hierarchical partition $\{\mathcal{B}_j\}_{j\in\NN_0}$ is then called \emph{area-regular}.
Once such a partition is given, we can then build our  Haar tight frame on $\sph$ directly from Theorem~\ref{Thm2.2}, which we name it \emph{area-regular spherical Haar tight framelets}.

The establishment of the hierarchical partition is made through a bijective mapping:
$ T:[-1, 1]\times[-1, 1] \rightarrow R^{(1)}\subset \mathbb{S}^2 $ defined by $T(x,y) = \frac{(x, y,1)}{\sqrt{x^2+y^2+1}}$.
%\begin{align*}
%	T(x,y) = \frac{(x, y,1)}{\sqrt{x^2+y^2+1}}.
%\end{align*}
See Figure~\ref{fig:T} for the illustration. Furthermore, one can verify that for any measurable set $ E \subseteq [-1, 1]\times[-1, 1] $, $\left|T(E)\right|= \int_{E} \eta(x,y)dxdy$,
%\begin{align*}
%	\left|T(E)\right|= \int_{E} \eta(x,y)dxdy,
%\end{align*}
where $\eta(x,y)=\frac 1 2(x^2+y^2+\frac 1 4 )^{-3/2}$ and $T(E)$ is the image of $E$ under $T$.

\begin{figure}[H]
	\centering
	\includegraphics[width=6cm]{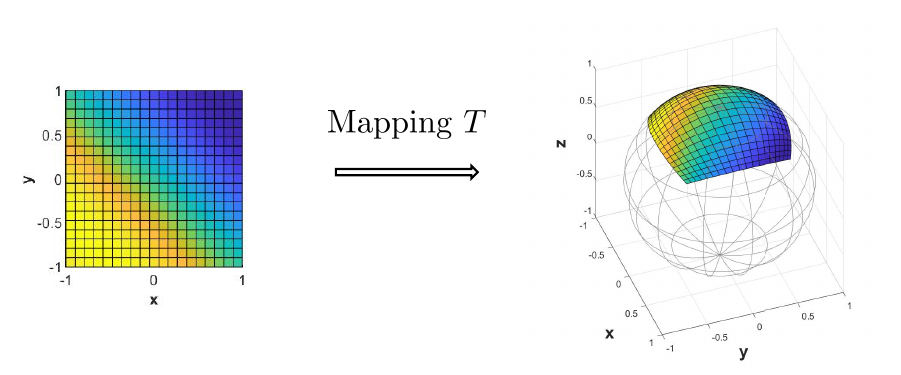}
	\caption{Visualization of mapping $T$
which maps a square to a spherical cap.}
	\label{fig:T}
\end{figure}

For each $R^{(1)}_{\vec v}$, we  build it in the form of
\[
R^{(1)}_{\vec v}=T\left([x^{\vec v}_{l},x^{\vec v}_{r}]\times [y^{\vec v}_{b},y^{\vec v}_{t}]\right)
\]
with $x^{\vec v}_{l}, x^{\vec v}_{r}, y^{\vec v}_{b},y^{\vec v}_{t}$ being in $[-1,1]\times [-1,1]$ indicating the left, right, bottom, top boundaries of a sub-block. The other $R_{\vec v}^{(s)}, s=2,\ldots,6$ are reflection versions of $R_{\vec v}^{(1)}$.   The following algorithm  define
\[
x^{\vec v}_{l},\, x^{\vec v}_{r},\, y^{\vec v}_{b},\, y^{\vec v}_{t}\quad
\forall \vec v \in [4]^j, \, j\in\NN
\]
iteratively.

\begin{enumerate}
	\item[(1)] Initially, we set
	\begin{eqnarray*}
		&x^{(1)}_{l}&= x^{(3)}_{l}=-1,\quad x^{(2)}_{r}=x^{(4)}_{r}=1,\\
		&x^{(1)}_{r}&=x^{(3)}_{r}=x^{(2)}_{l}=x^{(4)}_{l}=0\\
		&y^{(1)}_{b}&= y^{(2)}_{b}=-1,\quad y^{(3)}_{t}=y^{(4)}_{t}=1,\\
		&y^{(1)}_{t}&=y^{(3)}_{b}=y^{(2)}_{t}=y^{(4)}_{b}=0.
	\end{eqnarray*}
	
	\item[(2)]
	Given $x^{\vec v}_{l}, x^{\vec v}_{r}, y^{\vec v}_{b},y^{\vec v}_{t}$ and $\vec v\in [4]^k$ for $k\ge 1$, iteratively we define
	\begin{eqnarray*}
		&x^{(\vec v, 1)}_{l}&= x^{(\vec v,3)}_{l}=x^{\vec v}_{l},\\
		&x^{(\vec v, 2)}_{r}&=x^{(\vec v, 4)}_{r}=x^{\vec v}_{r},\\
		&x^{(\vec v,1)}_{r}&=x^{(\vec v,3)}_{r}=x^{(\vec v,2)}_{l}=x^{(\vec v,4)}_{l}=c,
	\end{eqnarray*}
	where $c$ is the solution such that
	\[\int_{y^{\vec v}_{b}}^{y^{\vec v}_{t}} \int_{x^{\vec v}_{l}}^c \eta(x,y)dxdy=\f{4\pi}6\f 1{4^{k}}\f 1 2.\]
	Furthermore, we define
	\begin{eqnarray*}
		y^{(\vec v,1)}_{b}&=&y^{(\vec v,2)}_{b}=y^{\vec v}_{b},\quad y^{(\vec v,3)}_{t}=y^{(\vec v,4)}_{t}=y^{\vec v}_{t},\\
		y^{(\vec v,1)}_{t}&=&y^{(\vec v,3)}_{b}=d_1,\quad y^{(\vec v,2)}_{t}=y^{(\vec v,4)}_{b}=d_2,
	\end{eqnarray*}
	where $d_1, d_2$ are the solutions such that
	\begin{align*}
		\int_{y^{\vec v}_1}^{d_1}\int_{x^{(\vec v,1)}_{l}}^{x^{(\vec v,1)}_{r}}\eta(x,y)dxdy &=\f{4\pi}6\f 1 {4^{k+1}}\quad  ,\\
		\quad \int_{y^{\vec v}_{b}}^{d_2}\int_{x^{(\vec v,2)}_{l}}^{x^{(\vec v,2)}_{r}}\eta(x,y)dxdy &=\f{4\pi}6\f 1 {4^{k+1}}.
	\end{align*}
\end{enumerate}

\begin{figure}[H]
	\centering
	\includegraphics[width=3in]{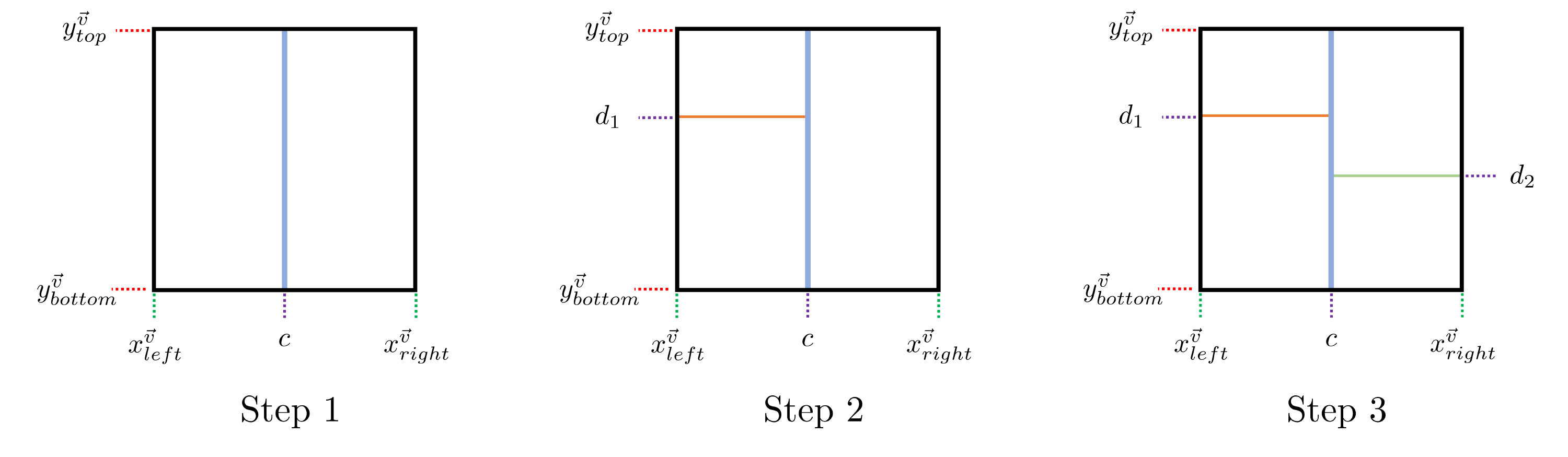}
	\caption{Visualization of our partition algorithm for calculating $c$, $d_1$ and $d_2$ at each iteration.}
	\label{apartition-square}
\end{figure}
Figure \ref{apartition-square} illustrates our algorithm. For each sub-block in each partition level, we first find $c \in [x_{l}^{\vec v}, x_{r}^{\vec v}]$ such that integration domain is divided into two areas that makes sub-blocks projected by the mapping $T$ are area-equal. Then $d_1, d_2 \in [y_{b}^{\vec v}, y_{t}^{\vec v}]$
are found so that the whole integration domain is divided into four spherical-area-equal domains. One should notice that $c$ does not need to be the midpoint and $d_1$, $d_2$ are not necessarily of the same value. It is worthwhile to point out that the above algorithm can be implemented efficiently by noting that
\begin{align*}
	& \int_{\beta_1}^{\beta_2}\int_{\alpha_1}^{\alpha_2} \eta(x,y)dxdy\\
= &z(\alpha_2,\beta_2) - z(\alpha_1,\beta_2) - z(\alpha_2,\beta_1) + z(\alpha_1,\beta_1),
\end{align*}
where
$	z(\alpha,\beta) = \arctan\left\{\frac{\alpha\beta}{\sqrt{\alpha^2+\beta^2+1} }\right\}$
for $\alpha,\beta\in[-1,1]$.
Once we have $R^{(1)}_{\vec v}$, we can obtain $R^{(s)}_{\vec v}$ for $s=2,\ldots,6$ by using reflections. More precisely, let $\tau_{\b w}$ to be the reflection about a vector $\b w\in\RR^3$, that is, $\tau_{\b w} \b x=\b x-\f{\b x\cdot \b w}{\|\b w\|}$. Let $\{e_1, e_2, e_3\}$ be the standard basis for $\RR^3$. Then setting $R^{(2)}_{\vec v}=\tau_{e_1-e_2}R^{(1)}_{\vec v}, R^{(3)}_{\vec v}=\tau_{e_1+e_2}R^{(1)}_{\vec v}, R^{(4)}_{\vec v}=\tau_{e_2-e_3}R^{(1)}_{\vec v}, R^{(5)}_{\vec v}=\tau_{e_2-e_3}R^{(1)}_{\vec v}, R^{(6)}_{\vec v}=\tau_{e_2}R^{(1)}_{\vec v}$ for all $\vec v\in [4]^j$, we could get an area-regular hierarchical partition with $\CB_j=\left\{R^{(s)}_{\vec v}\right\}_{\vec v\in [4]^j, s\in [6]}$. Figure \ref{partition-S2} illustrates $\mathcal B_{0}$ to $\mathcal B_{3}$.

\begin{figure}[h]
	\centering
	\subfigure[$\mathcal B_{0}$]{\includegraphics[width=1.8cm]{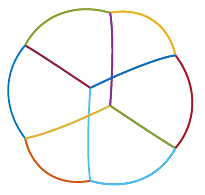}}
	\subfigure[$\mathcal B_{1}$]{\includegraphics[width=1.8cm]{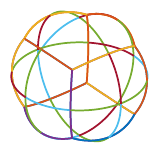}}
	\subfigure[$\mathcal B_{2}$]{\includegraphics[width=1.8cm]{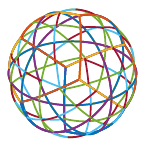}}
	\subfigure[$\mathcal B_{3}$]{\includegraphics[width=1.8cm]{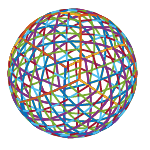}}
	\\
	\subfigure[$ R^{(1)}$]{\includegraphics[width=1.8cm]{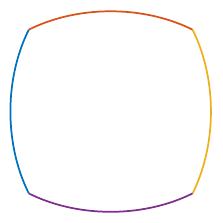}}
	\subfigure[$ R^{(1)}_{(\cdot)}$]{\includegraphics[width=1.8cm]{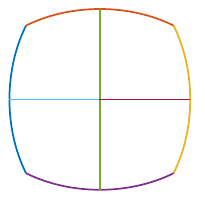}}
	\subfigure[$ R^{(1)}_{(\cdot,\cdot)}$]{\includegraphics[width=1.8cm]{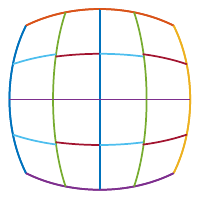}}
	\subfigure[$ R^{(1)}_{(\cdot,\cdot,\cdot)}$]{\includegraphics[width=1.8cm]{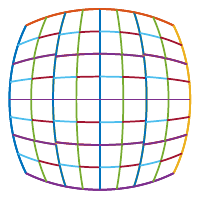}}
	\caption{(a)-(d) Sphere partition $\mathcal B_{0}$ to $\mathcal B_{3}$;(e)-(h) the upper surface partition. }\label{partition-S2}
\end{figure}

With the above partition, for all $j\in \NN_0$, $\vec v\in [4]^j$ and $s\in[6]$, we define $\phi^{(s)}_{\vec v}= \f{1}{\sqrt{R^{(s)}_{\vec v}}}\chi_{R^{(s)}_{\vec v}}$ and the Haar framelets to be
\[
\begin{small}
\begin{aligned}
		&\psi^{(s)}_{(\vec v,1)}&= \f{1}{\sqrt{2}} (\phi^{(s)}_{(\vec v,1)}-\phi^{(s)}_{(\vec v,2)}),\quad		\psi^{(s)}_{(\vec v,2)}= \f{1}{\sqrt{2}} (\phi^{(s)}_{(\vec v,1)}-\phi^{(s)}_{(\vec v,3)}),\\
		&\psi^{(s)}_{(\vec v,3)}&= \f{1}{\sqrt{2}} (\phi^{(s)}_{(\vec v,1)}-\phi^{(s)}_{(\vec v,4)}),\quad		\psi^{(s)}_{(\vec v,4)}= \f{1}{\sqrt{2}} (\phi^{(s)}_{(\vec v,2)}-\phi^{(s)}_{(\vec v,3)}),\\
	&\psi^{(s)}_{(\vec v,5)}&= \f{1}{\sqrt{2}} (\phi^{(s)}_{(\vec v,2)}-\phi^{(s)}_{(\vec v,4)}),\quad		\psi^{(s)}_{(\vec v,6)}= \f{1}{\sqrt{2}} (\phi^{(s)}_{(\vec v,3)}-\phi^{(s)}_{(\vec v,4)}),
	\end{aligned}
\end{small}
\]
which is associated with matrix
\[
\b A=\f 1 {\sqrt 2}
\begin{bmatrix}
	1 & -1 & 0 & 0 \\
	1 & 0 & -1 & 0 \\
	1 & 0 & 0 & -1 \\
	0 & 1 & -1 & 0 \\
	0 & 1 & 0 & -1 \\
	0 & 0 & 1 & -1
\end{bmatrix},
\]
$\b p=\f{1}{2}[1,1,1,1]^\top$,
and
$$\b Q=
\begin{bmatrix}
	\f{1}{2} & \f{\sqrt{2}}{4}  & \f{\sqrt{2}}{4}  & \f{\sqrt{2}}{4}  & 0  & 0  & 0 \\
	\f{1}{2} & -\f{\sqrt{2}}{4} & 0  & 0  & \f{\sqrt{2}}{4}  & \f{\sqrt{2}}{4}  & 0 \\
	\f{1}{2} & 0  & -\f{\sqrt{2}}{4} & 0  & -\f{\sqrt{2}}{4} & 0  & \f{\sqrt{2}}{4}\\
	\f{1}{2} & 0  & 0  & -\f{\sqrt{2}}{4} & 0  & -\f{\sqrt{2}}{4} & -\f{\sqrt{2}}{4}
\end{bmatrix}
$$
 in Theorem~\ref{Thm2.2}. The system
$\CF_L(\{\CB_j\}_{j\in\NN_0}):=\{\phi^{(s)}_{\vec v},s\in[6]\}_{\vec v\in \Lambda_L}\cup\{\psi^{(s)}_{(\vec v, \ell)}, s,\ell\in[6]\}_{j\ge L,\vec v\in [4]^j}$
 forms an area-regular spherical  Haar tight frame with frame bound $2$ for $L_2(\sph)$.

\section{Denoising experiments}
\label{sec:exp}

Based on the construction of area-regular spherical Haar tight framelets and the decomposition-reconstruction algorithm above, in this section we exploit two different kinds of methods to the denoising problem for signals on the 2D-sphere: one is the classical thresholding techniques on the framelet coefficient domains, and the other is to train a CNN model for spherical signal denoising.

\subsection{Thresholding methods}
Under the hierarchical partition constructed in section~\ref{sec:sphere}, given a  spherical signal, we first sample the signal to  the space $\CV_J$ for some appropriate $J\in\NN$.
%and denote it as $f_{\CV_J}$.
Then applying the decomposition algorithm in section~\ref{sec:constr}
%with
%$$\b A=\f 1 {\sqrt 2}
%\begin{bmatrix}
%	1 & -1 & 0 & 0 \\
%	1 & 0 & -1 & 0 \\
%	1 & 0 & 0 & -1 \\
%	0 & 1 & -1 & 0 \\
%	0 & 1 & 0 & -1 \\
%	0 & 0 & 1 & -1
%\end{bmatrix},
%$$
%and $\b p=\f{1}{2}(1,1,1,1)^\top$
to obtain the lowpass signals $c^{(L)}_{\vec v}$ and highpass signals $d_{(\vec v,i)}^{(j)}$ for $\vec v\in\Lambda_j$ and $i\in[6]$, which correspond to the transform coefficients for $\CV_L$ and $\CW_j$, $j=L,\ldots, J$ with $L<J$. In the following, we simply call $J-L$ in descriptions and tables as ``level'' unless specified.
We then adapt three thresholding techniques to our denoising problem and provide a comparison on their performance.
The first one is the so-called \emph{soft thresholding},  by updating
\[
d^{(j)}_{(\vec v,i)} \leftarrow \operatorname{sgn}
\left(d^{(j)}_{(\vec v,i)}\right)
\max\left\{|d^{(j)}_{(\vec v,i)}|-\lambda_s,0\right\},
\]
where $\lambda_s$ is usually set independent of the position $(\vec v,i)$. Other two methods are based on the soft thresholding, which are the  \emph{local-soft thresholding} and the \emph{bivariate shrinkage thresholding}. The local-soft thresholding  designs the thresholding value by considering local information of each coefficient. Denote $\s_b := \s \|b\|$ and $ \s_i= \sqrt{(\hat{\s_i}^2 - \s_b^2)_{+}} $ with $ \hat{\s}_i^2 = \f{1}{\# W_i} \sum_{e_j \in W_i} |e_j|^2 $
	where $\s$ is the noise variance, $\|b\|$ is the filter 2-norm,  $\# W_i$ measures the size of the window $W_i$ centering at $(\vec v,i)$, and $e_j$ are  coefficients inside such a window. Then local-soft thresholding value replaces $\lambda_s$ in the soft thresdholding by $\lambda_{ls}:=\f{r\s_b^2}{\s_i}$ for some positive constant $r$. On the other hand, the bivariate shrinkage set the thresholding value by considering parent coefficient information as
$\lambda_{bs} := \f{\lambda_{ls}}{ \sqrt{1 + \left|{d^{(j-1)}_{\vec v }}/{d^{(j)}_{\vec v, i}}\right|^2}}$. See \cite{HZ, BS} for more details.

\subsection{Denoising for different dataset on the sphere}
Applying the above methods, some experiments are conducted on ETOPO \cite{ETOPO}, 2D gray images, and MNIST, CIRFAR10, Caltech101, in which hyperparameters are chosen specifically. For soft thresholding, $\lambda_s$ is chosen to be $0.9 \times rate \times f_{\max}$, where $f_{\max}$ is the largest absolute value of signal $f$. For local soft and bivariate thresholding, widow size $W_i$ and $r$ are taken to be $2$ and $0.3$, respectively. We employ peak signal-to-noise ratio(PSNR) in unit ${\rm  dB}$ to evaluate the denoising performance, defined by $ \text{PSNR} = 10 \log_{10} \f{ f_{\max}^2 }{ \text{MSE} } $, where $\text{MSE}$ measures the mean square error between noised (reconstruction) data and ground truth.  We use the additive Gaussian noise with varying standard deviation $\sigma=rate \times f_{\max}$, for $rate=0.05, 0.1, 0.2, 0.5$, which is added directly to the spherical images. So for the case $f_{\max}=255$ (gray scale), the noise deviations correspond $\sigma=12.75, 25.5, 51$ and $127.5$, respectively.

\begin{figure}[http]
	\centering
	%\includegraphics[width=3in]{sphe_8}\\
	\subfigure[$26.02$dB]{\includegraphics[width=2cm]{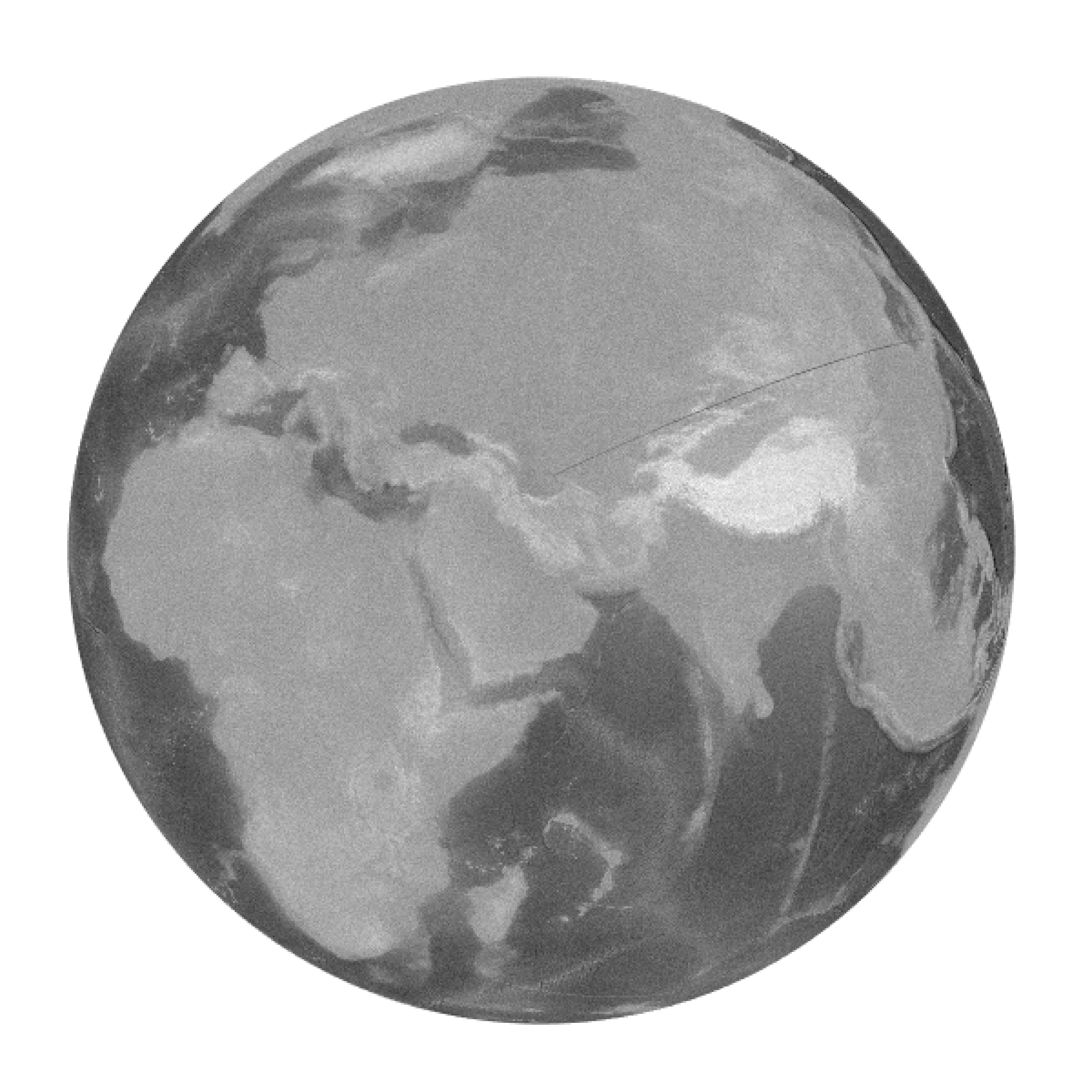}}
	\subfigure[$19.99$dB]{\includegraphics[width=2cm]{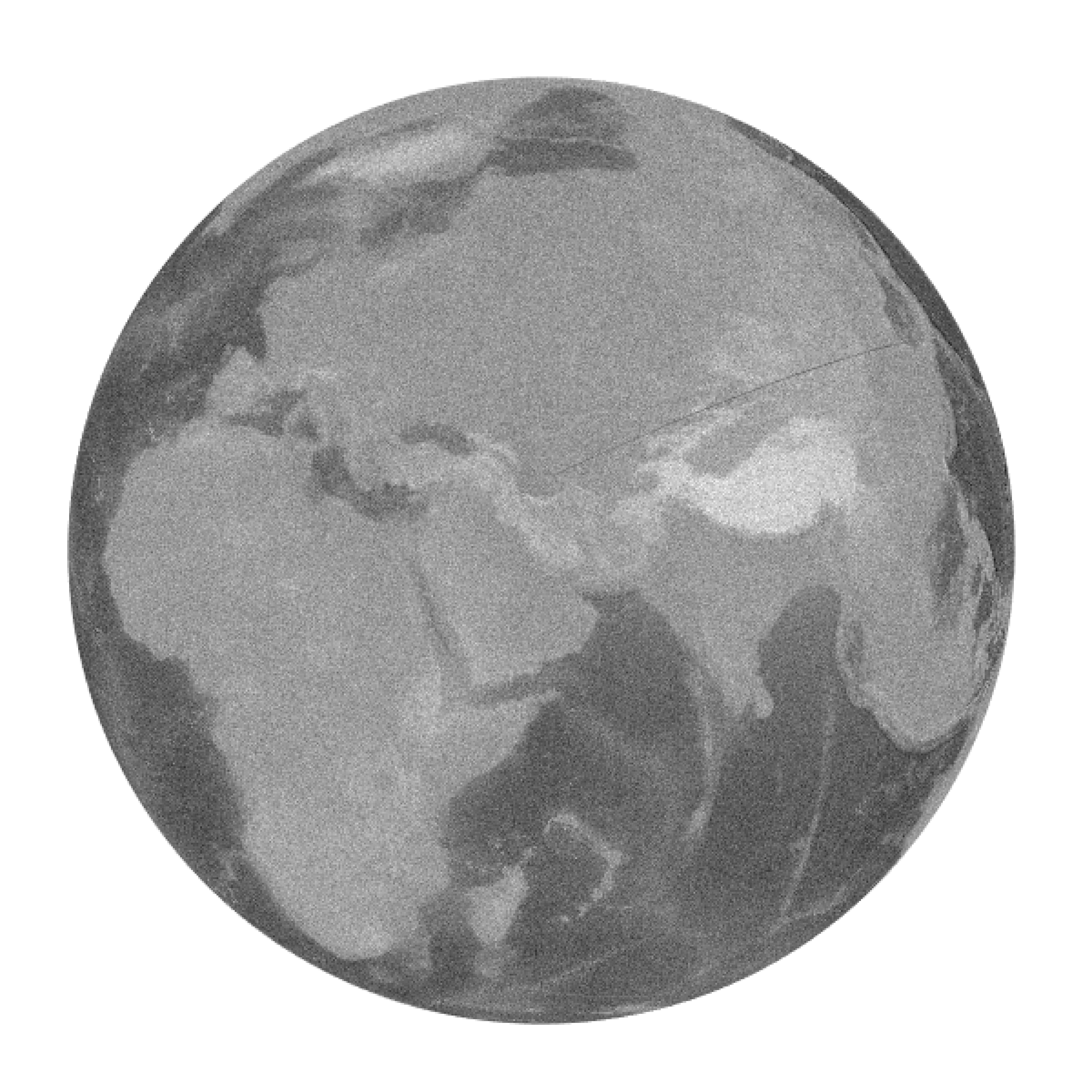}}
	\subfigure[$13.97$dB]{\includegraphics[width=2cm]{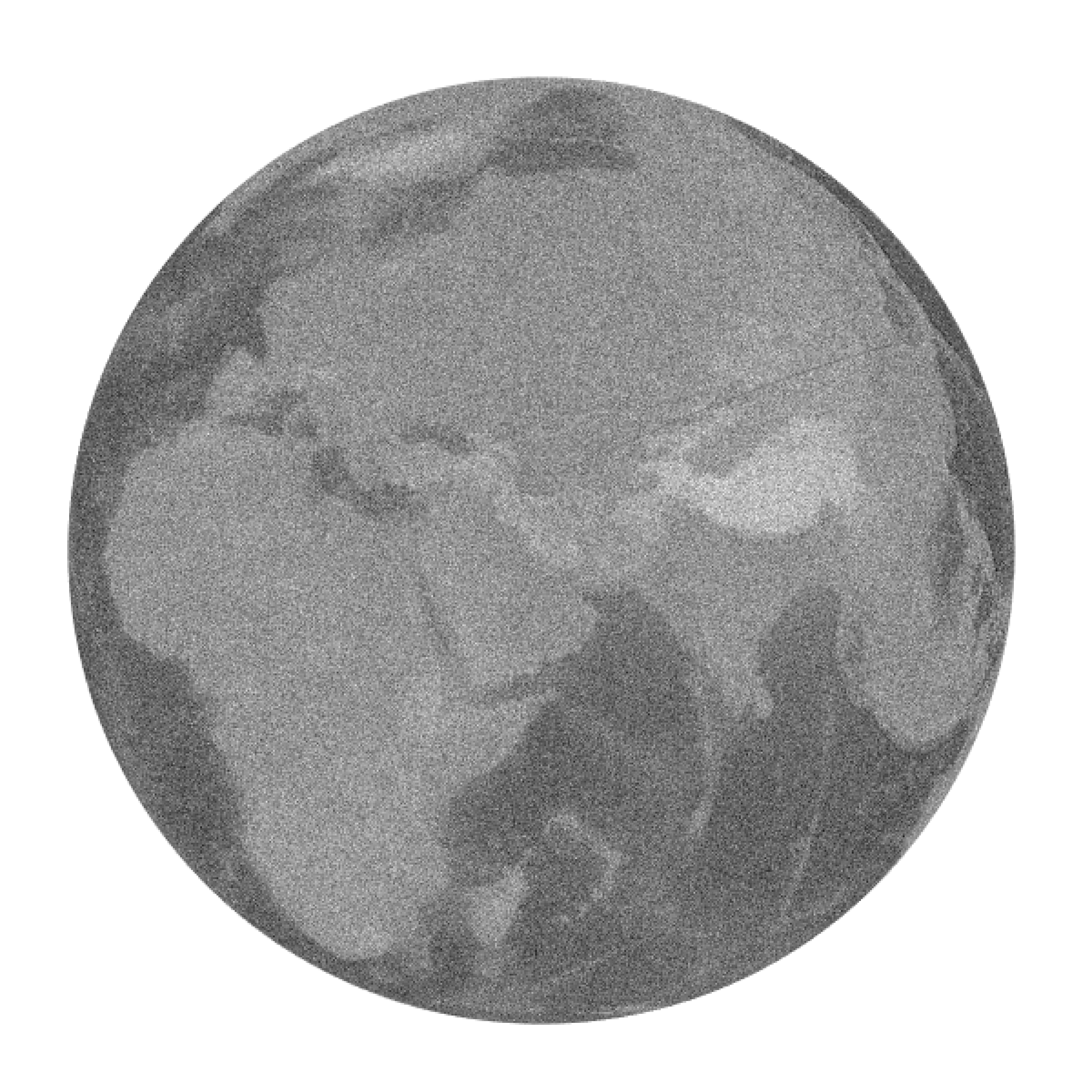}}
	\subfigure[$6.02$dB]{\includegraphics[width=2cm]{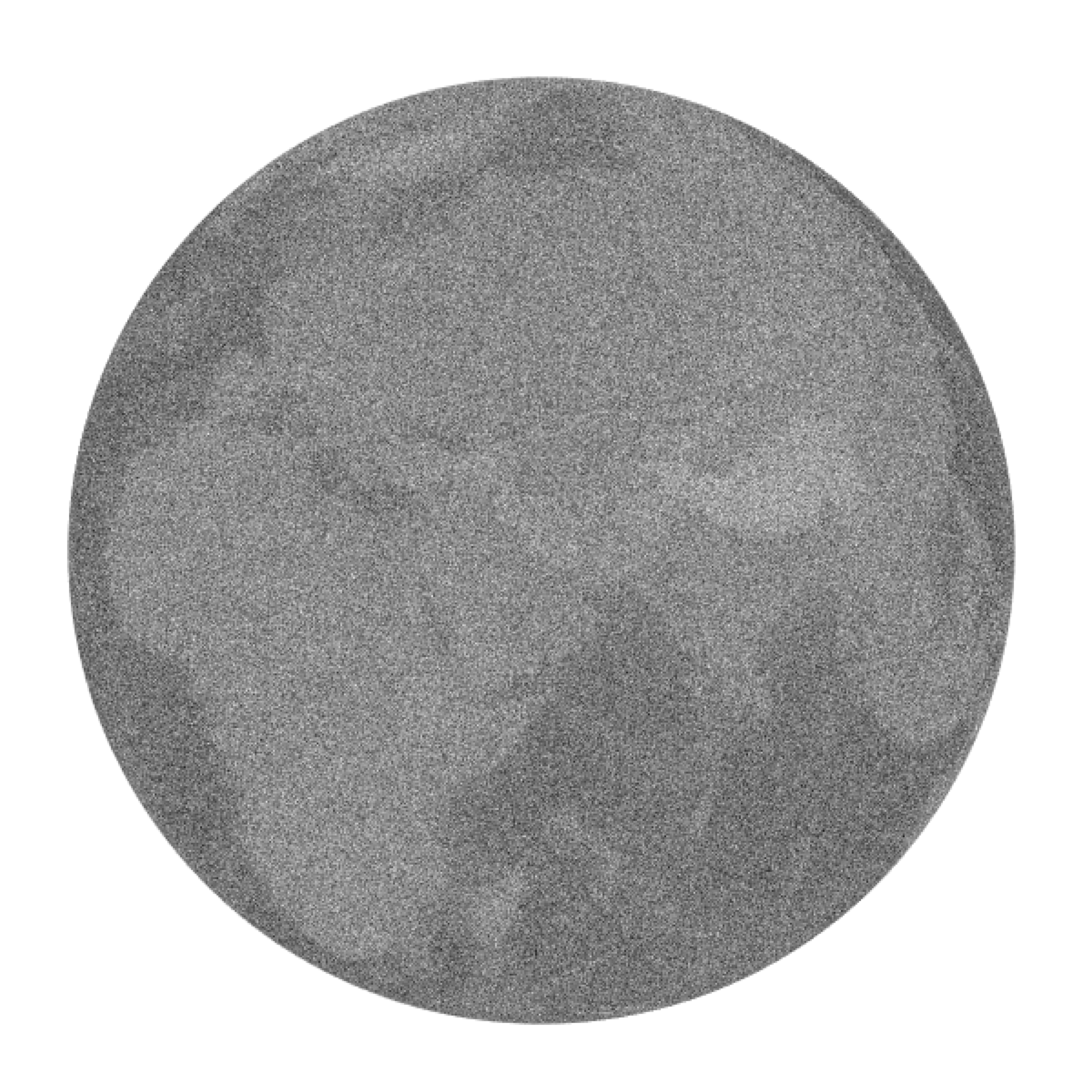}}
	\\
	\centering
	\subfigure[$33.50$dB]{\includegraphics[width=2cm]{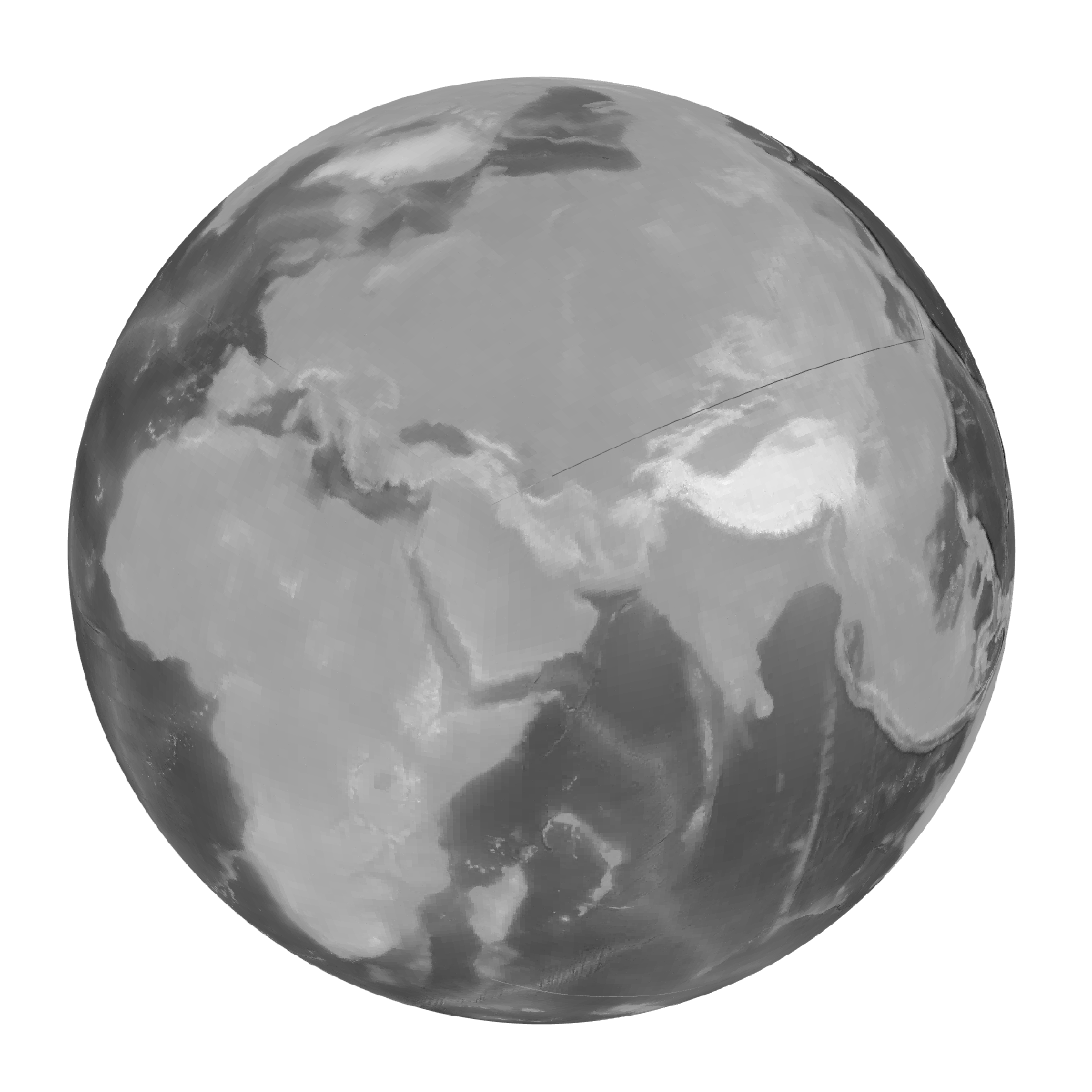}}
	\subfigure[$31.09$dB]{\includegraphics[width=2cm]{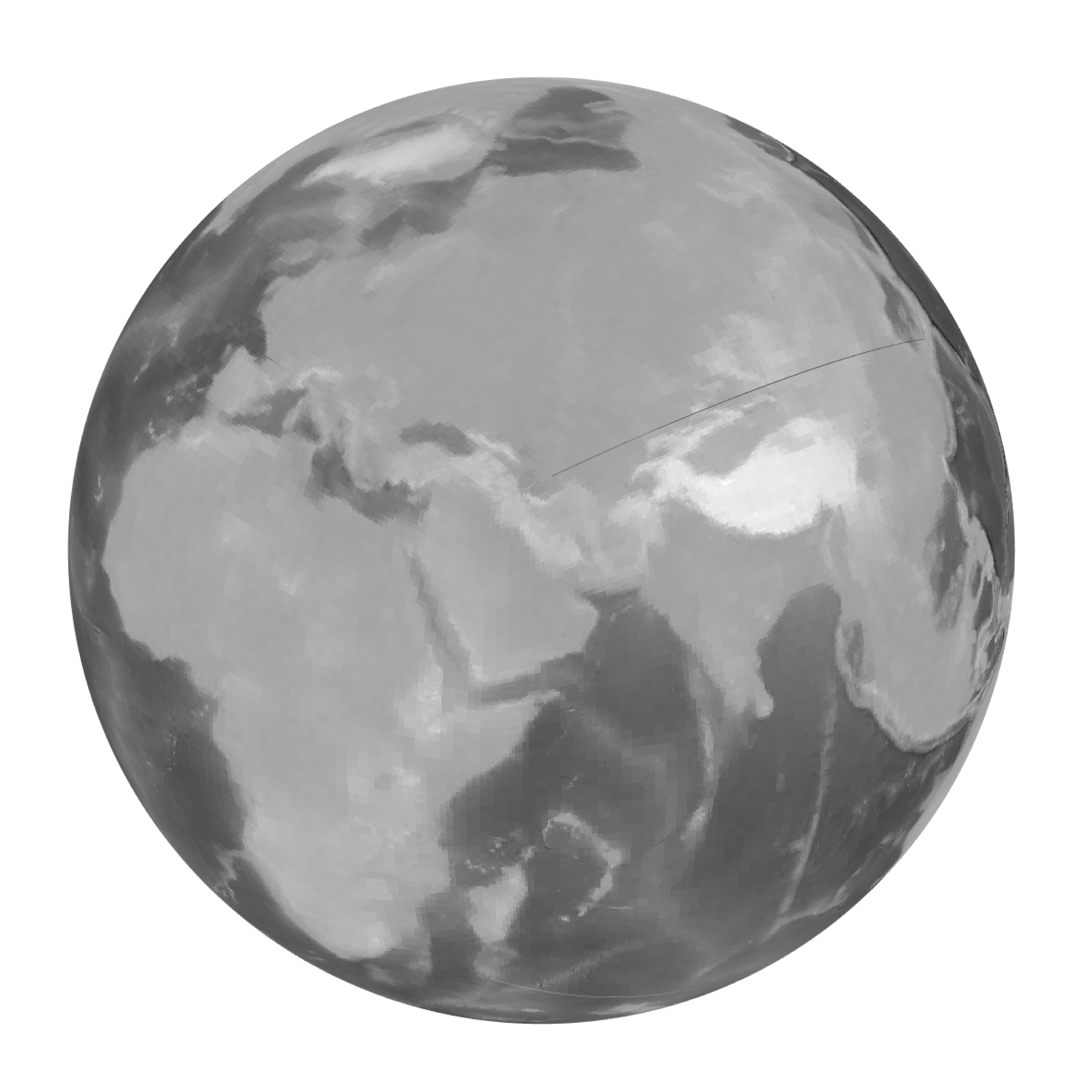}}
	\subfigure[$28.86$dB]{\includegraphics[width=2cm]{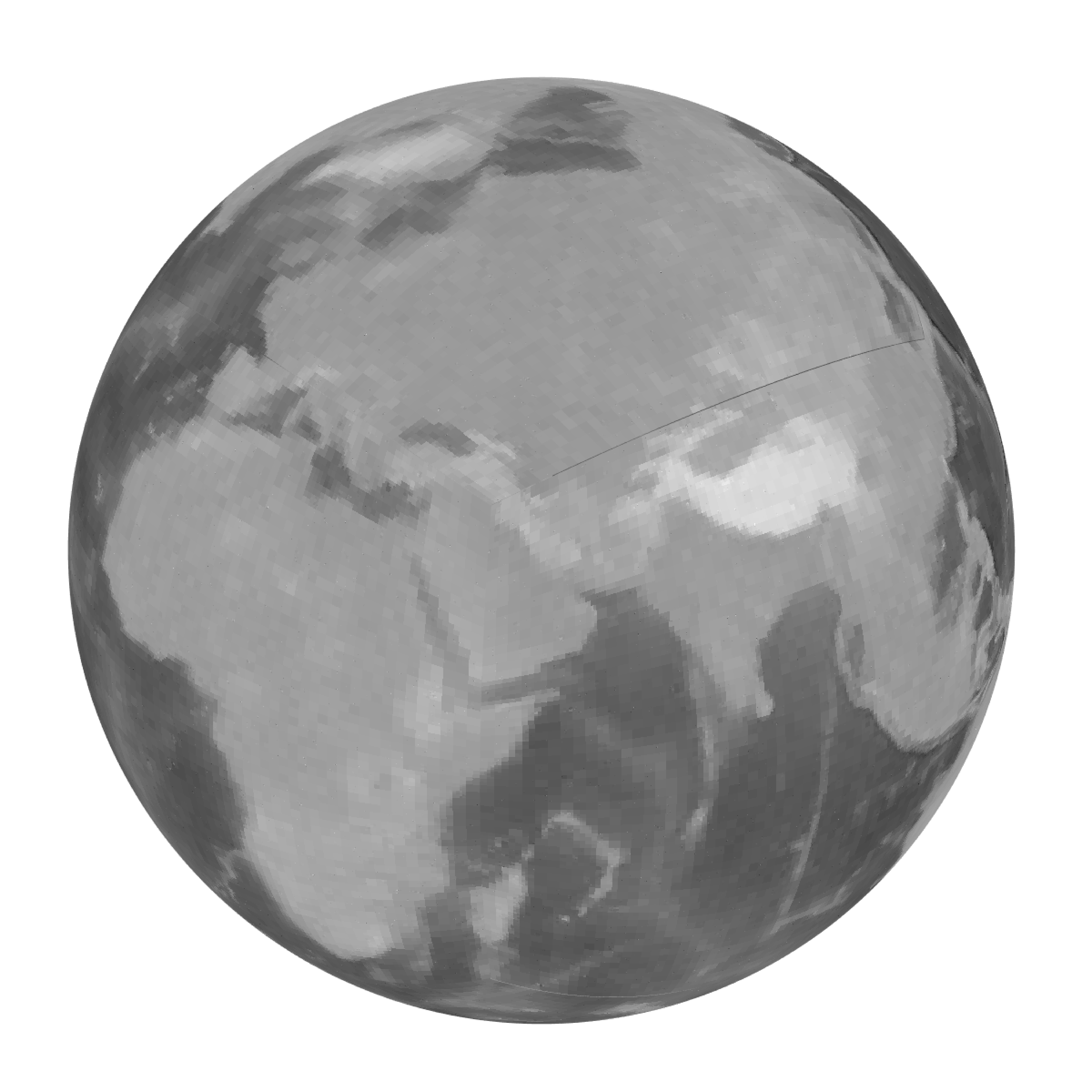}}
	\subfigure[$25.30$dB]{\includegraphics[width=2cm]{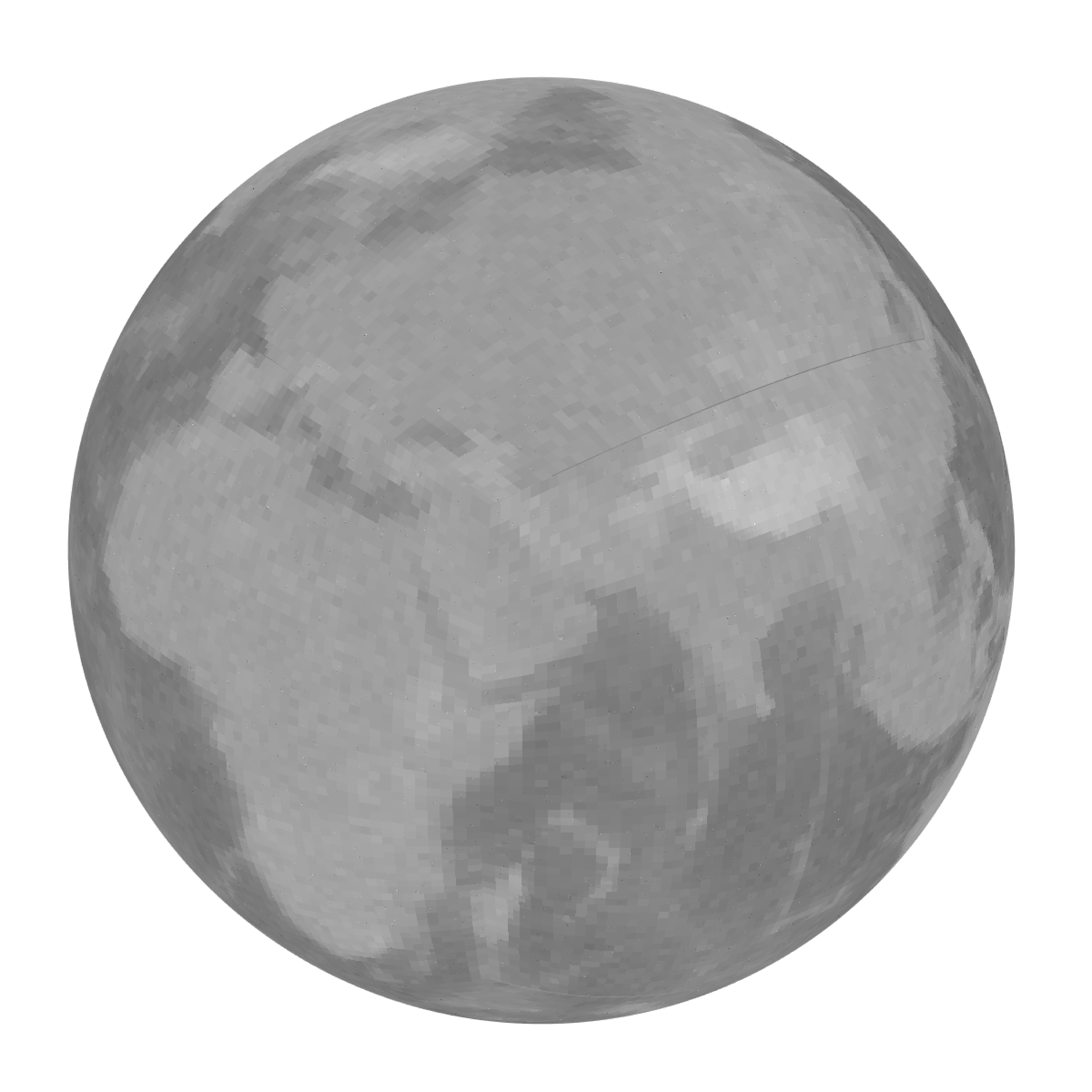}}
	\caption{ETOPO 4-level denoising experiment with bivariate thresholding. (a)-(d) are corrupted images and (e)-(h) are corresponding denoised images. PSNR are shown above.}
		\label{ETOPO_DENOISE}
\end{figure}

\begin{table}[http]
	\centering
	\footnotesize
	\setlength\tabcolsep{1pt}
	\caption{ETOPO denoising results}
	\label{ETOPO_table}
	%\renewcommand\tabcolsep{1.5pt}
	%\renewcommand\arraystretch{1.5}
	%\tiny
	\begin{tabular}{|c|c|c|c|c|c|}
		\hline
		\textbf{Level}              & \textbf{Rate} & \textbf{0.05(26)} & \textbf{0.1(20)} & \textbf{0.2(14)}  & \textbf{0.5(6)} \\ \hline
		\multirow{3}{*}{\textbf{3}} & soft          & 33.18       & 30.47        & 27.05           & 20.76       \\ \cline{2-6}
		& localsoft    & 33.37        & 30.83      & 28.07           & 22.78       \\ \cline{2-6}
		& bivariate     & 33.48         & 30.90        & 28.09            & 22.76       \\ \hline
		\multirow{3}{*}{\textbf{4}} & soft          & 33.00         & 30.44        & 27.53          & 22.21      \\ \cline{2-6}
		& localsoft    & 33.37      & 30.99       & 28.80          & 25.32       \\ \cline{2-6}
		& bivariate     & \textbf{33.50(7.48$\uparrow$)}          & \textbf{31.09(11.09$\uparrow$) }       & 28.86      & 25.30       \\ \hline
		\multirow{3}{*}{\textbf{5}} & soft          & 32.89        & 30.28     & 27.40        & 22.36      \\ \cline{2-6}
		& localsoft    & 33.36        & 30.98      & 28.82          & 25.76     \\ \cline{2-6}
		& bivariate     & 33.49         & 31.08    & \textbf{28.90(14.93$\uparrow$)}             & \textbf{25.77(19.76$\uparrow$) }      \\ \hline
	\end{tabular}
\end{table}

In the ETOPO experiment, we conduct our algorithm and methods to spherical images which contain information of land topography and ocean bathymetry \cite{ETOPO}.  $5400 \times 10800$ points are sampled from the original image  %shown in Figure \ref{ETOPO_ORIGINAL_MAP}
and resampled to the partition with respect to  $\CV_{10}$, which contains about
$6.3 \times 10^6$ points. Table \ref{ETOPO_table} shows denoising results with different decomposition levels, in which and other tables the first row represents noise rate and PSNR after adding noise, and the remaining rows represent PSNR after denoising and the improved value in PSNR. Performance is greatly improved as decomposition level increasing. %but with more than three times decomposition the quality of the images is almost unchanged.
Even though the noise almost  ruin the original image when $rate=0.5$, the proposed denoising improves PSNR to near $20$ with bivariate thresholding. Figure~\ref{ETOPO_DENOISE} shows a plot of corrupted images and denoised images.

\begin{figure}[http]
	\centering
	\includegraphics[width=3.5in]{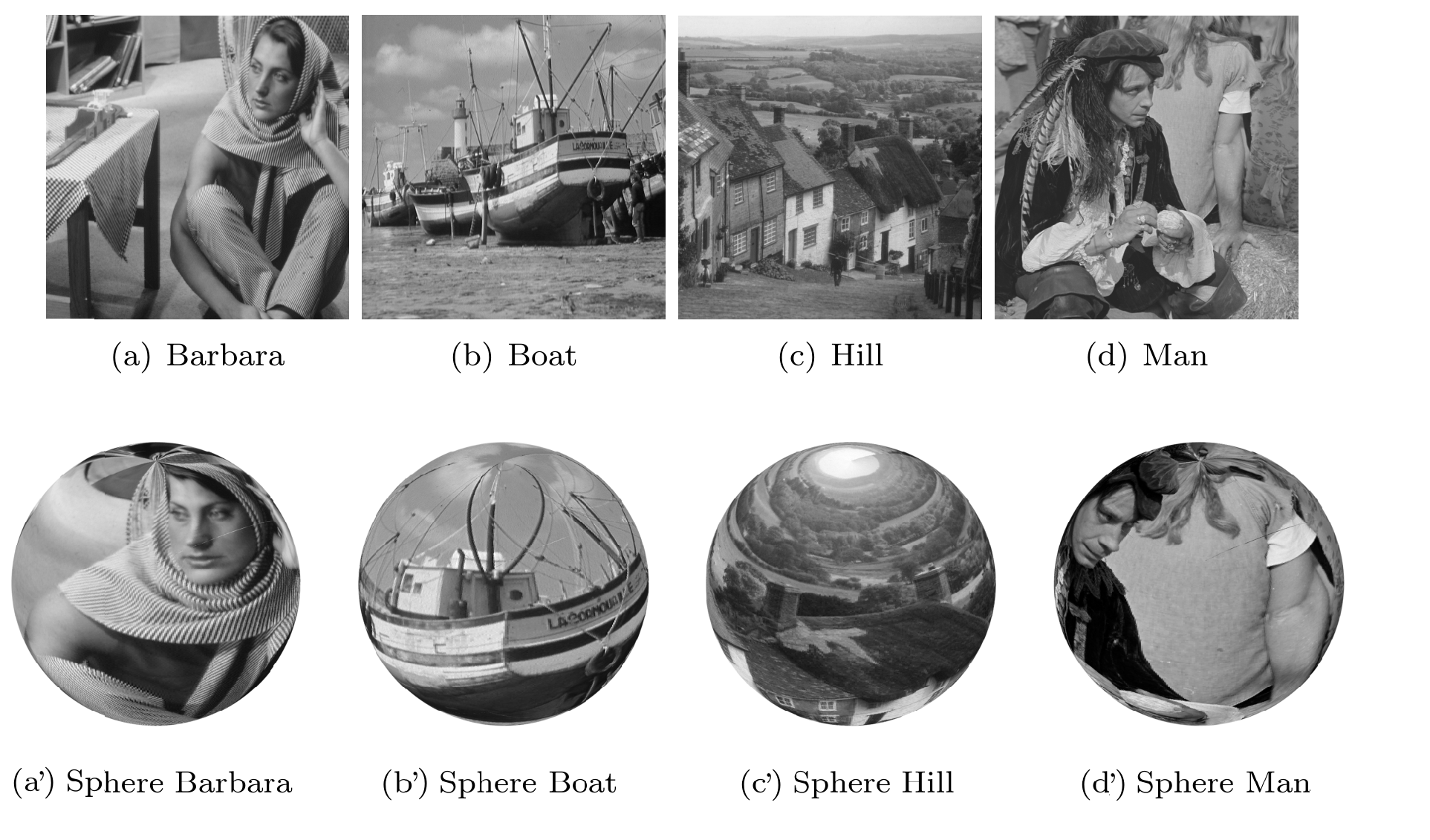}\\
	\caption{2D gray images}\label{sph2Dimg}
\end{figure}

In the 2D gray image experiment, we consider spherical maps produced by several $512\times512$ classical images called as ``Barbara'',  ``Boat'',  ``Hill'',  and ``Man''. Original images and corresponding spherical images which are resampled to the partition with respect to   $\CV_8$. Figure~\ref{sph2Dimg} shows original images and corresponding sampled spherical images.
We conducted the denoising by thresholding for decomposition up to 4 levels ($J=8,L=4$). The experiment result shows that PSNRs are significantly increased, see Table~\ref{level4}.

\begin{table}[http]
	\centering
	%\tiny
	\footnotesize
	\setlength\tabcolsep{1pt}
	\caption{2D gray images threshold results}
	\begin{tabular}{|c|c|c|c|c|c|}
		\hline
		\textbf{Dataset}           & \textbf{Rate} & \textbf{0.05(26)} & \textbf{0.1(20)} & \textbf{0.2(14)}  & \textbf{0.5(6)} \\ \hline
		\multirow{3}{*}{\textbf{\tiny Barbara}} & soft        & 27.08         & 24.00         & 21.71         & 18.77       \\ \cline{2-6}
		& localsoft  & 28.42         & 24.37         & 21.80         & 19.53       \\ \cline{2-6}
		& bivariate     & \textbf{28.64(2.64$\uparrow$) }         & \textbf{24.52(4.53$\uparrow$)}         &\textbf{ 21.89(7.92$\uparrow$) }           & \textbf{19.55(13.55$\uparrow$)  }     \\ \hline
		\multirow{3}{*}{\textbf{\tiny Boat}}    & soft        & 28.58         & 25.61         & 22.94         & 19.26     \\ \cline{2-6}
		& localsoft  & 29.45         & 26.09         & 23.33         & 20.27      \\ \cline{2-6}
		& bivariate     & \textbf{29.64(3.63$\uparrow$) }         & \textbf{26.24(6.25$\uparrow$) }        & \textbf{23.45(9.48$\uparrow$)  }               & \textbf{20.32(14.32$\uparrow$)  }     \\ \hline
		\multirow{3}{*}{\textbf{\tiny Hill}}    & soft        & 28.36         & 25.76         & 23.46         & 19.95       \\ \cline{2-6}
		& localsoft  & 28.85         & 25.98         & 23.74         & 21.18     \\ \cline{2-6}
		& bivariate     & \textbf{29.02(3.01$\uparrow$)   }       & \textbf{26.09(6.11$\uparrow$) }        & \textbf{23.82(9.86$\uparrow$) }              & \textbf{21.20(15.20$\uparrow$)    }   \\ \hline
		\multirow{3}{*}{\textbf{\tiny Man}}     & soft        & 28.92         & 26.03         & 23.49         & 19.78      \\ \cline{2-6}
		& localsoft  & 29.64         & 26.35         & 23.83         & 20.97      \\ \cline{2-6}
		& bivariate   & \textbf{29.79(3.78$\uparrow$) }   & \textbf{26.47(6.48$\uparrow$) }        & \textbf{23.92(9.95$\uparrow$) }               & \textbf{20.99(14.99$\uparrow$)    }   \\ \hline
	\end{tabular}
	\label{level4}
\end{table}

In the MNIST, CIFAR10 and Caltech101 experiment, we first convert 10000 images from MNIST \cite{MNIST}, 10000 images from CIFAR10 \cite{cifar10}, and 1000 images from Caltech101 \cite{caltech101} to be gray images and then produce their spherical image dataset under the partition of 2D-sphere in section II.
See Figure~\ref{fig:mnist} for the illustration.
%we take three spherical datasets: with 60000 images for training and  10000 for testing, CIFAR10 \cite{cifar10} with 50000 images for training and 10000 for testing and Caltech101 \cite{caltech101} with first 1000 images for testing and 7677 images for training. For CIFAR10 and Caltech101, we first convert them to be gray images and then produce spherical image dataset.
Precisely, images from MNIST are sampled by using the partition $\CB_4$, and  images from CIFAR10 and Caltech101 are sampled by using the partition $\CB_6$. The threshold denoising are conducted by decomposition and reconstruction up to 2 levels.
 %in low resolution, in this experiment we conduct the same denoising procedure as above but only take decomposition up to $2$ levels with $J = 4$. For images from CIFAR10 and Caltech101, we samples images to $J = 6$.
 Table \ref{mnist_denoising} lists the averaged PSNR  of above three datasets before and after denoising.% and our denoising algorithm works well on these CIFAR10 and Caltech101.

\begin{figure}[htpb]
	\centering
	\includegraphics[width=3.5in]{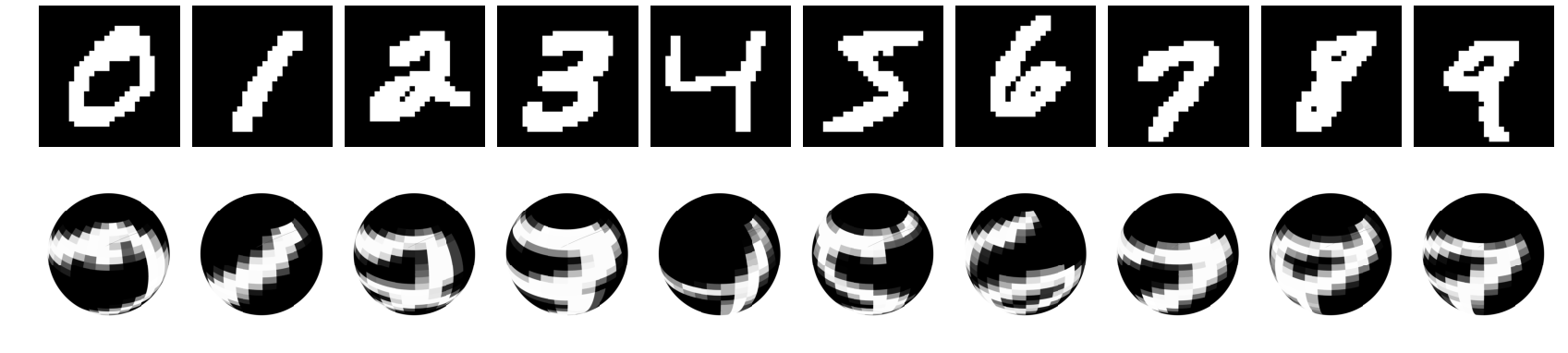}\\
	\caption{MNIST (top) and MNIST on the sphere (bottom).}
\label{fig:mnist}
\end{figure}

\begin{table}[http]
	\centering
	%\tiny
	%\footnotesize
	\scriptsize
	\setlength\tabcolsep{1pt}
	\caption{Threshold results for MNIST, CIFAR10 and Caltech101 .}
	\begin{tabular}{|c|c|c|c|c|c|c|}
		\hline
		\textbf{Dataset}                     & \textbf{Level}     & \textbf{Rate} & \textbf{0.05(26)} & \textbf{0.1(20)} & \textbf{0.2(14)}  & \textbf{0.5(6)}  \\ \hline
		\multirow{6}{*}{\textbf{MNIST}}      & \multirow{3}{*}{1} & soft          & 27.64  & 22.16 & 17.12  & 10.67 \\ \cline{3-7}
		&                    & local soft    & 28.87 & 22.66 & 16.80  & 10.69 \\ \cline{3-7}
		&                    & bivariate     & 28.84  & 22.73& 16.86  & 10.69 \\ \cline{2-7}
		& \multirow{3}{*}{2} & soft          & 27.71& 22.45 & \textbf{17.91(4.11$\uparrow$)}  & \textbf{13.30(7.45$\uparrow$)}  \\ \cline{3-7}
		&                    & local soft    & \textbf{29.95(4.11$\uparrow$)}  & 23.64 & 17.73  & 13.09 \\ \cline{3-7}
		&                    & bivariate     & 29.87 & \textbf{23.73(3.90$\uparrow$)} & 17.85 & 13.10  \\ \hline
		\multirow{6}{*}{\textbf{CIFAR10}}    & \multirow{3}{*}{1} & soft          & 26.01 & 22.09& 17.92 & 11.04  \\ \cline{3-7}
		&                    & local soft    & 26.38  & 21.82& 17.82  & 11.14  \\ \cline{3-7}
		&                    & bivariate     & \textbf{26.54(0.70$\uparrow$)}  & 21.89& 17.82 & 11.14  \\ \cline{2-7}
		& \multirow{3}{*}{2} & soft          & 25.49 & \textbf{22.16(2.33$\uparrow$)}& \textbf{19.44(5.64$\uparrow$)}  & 15.05(9.20$\uparrow$)  \\ \cline{3-7}
		&                    & local soft    & 26.32  & 21.95 & 19.08  & \textbf{15.32(9.47$\uparrow$) } \\ \cline{3-7}
		&                    & bivariate     & 26.53  & 22.10 & 19.12  & 15.30  \\ \hline
		\multirow{6}{*}{\textbf{Caltech101}} & \multirow{3}{*}{1} & soft          & 27.63 & 23.42 & 18.76  & 11.57  \\ \cline{3-7}
		&                    & local soft    & 27.89 & 23.28 & 18.73  & 11.70  \\ \cline{3-7}
		&                    & bivariate     & 27.99 & 23.32 & 18.73  & 11.70  \\ \cline{2-7}
		& \multirow{3}{*}{2} & soft          & 27.47  & 24.21 & \textbf{21.18(7.23$\uparrow$)} & 16.05\\ \cline{3-7}
		&                    & local soft    & 28.15 & 24.10& 21.01  & \textbf{16.48(10.48$\uparrow$)} \\ \cline{3-7}
		&                    & bivariate     & \textbf{28.31(2.34$\uparrow$) } & \textbf{24.21(4.25$\uparrow$) }& 21.04  & 16.47\\ \hline
	\end{tabular}
	 \label{mnist_denoising}
\end{table}

\subsection{A CNN model for spherical signal denoising}
%Aside with the classical thresholding methods, in this subsection we exploit to apply our spherical Haar framelet decomposition and reconstruction algorithms to construct convolutional neural networks (CNNs) for the denoising.
%denoise by employing convolutional neural networks (CNNs) model, in which the Haar decomposition and reconstruction algorithms are involved.
In this subsection, combining with our spherical Haar framelet decomposition and reconstruction algorithms we propose a convolutional neural network (CNN) for the denoising.
The architecture of our model as illustrated by Figure~\ref{nn}  comprises four ConvConvT cells in  which 2D convolution, transpose convolution, and decomposition and reconstruction by our Haar tight framelets are applied. %Figure \ref{nn} illustrates the structure.
All convolutions and transpose convolutions have kernel size of $3\times3$ or $2\times2$. Also, we  make cell 1 \& 2 share the same parameters. However, the input of Cell 2 is summation of two parts, one of them is the origin and the other is from reconstruction. In Cell 3 \& 4, we take addition of previous convolution layers to transpose convolution layers. As a bridge between Cell 1(2) and Cell 3(4), the Haar decomposition and reconstruction algorithms are conducted. Finally, we choose ReLU function to activate every convolution and transpose convolution except the last transpose convolution in Cell 1 \& 2. With our area-regular hierarchical partition, spherical signals are treated as matrices which preserve the neighbourhoods information, while the convolutions can be conducted same as those for regular 2D images.
%From the view of signal processing algorithm, ConvConvT cells can be regarded as having adaptive decomposition and reconstruction kernel with three decomposition level. Alternative conducting convolution and transpose convolution do not restrict the size of the inputs which is the similar to thresholding functions.  We add these characteristics of signal processing algorithm so as to enhance the generalization ability of proposed denoising neural network.
%
%As shown in Figure \ref{nn}, building blocks of our neural network only have $3$ parts: 2D convolution, transpose convolution and decomposition and reconstruction by tight haar framelets.

\begin{figure*}[htp]
	\centering
	\includegraphics[scale=0.5]{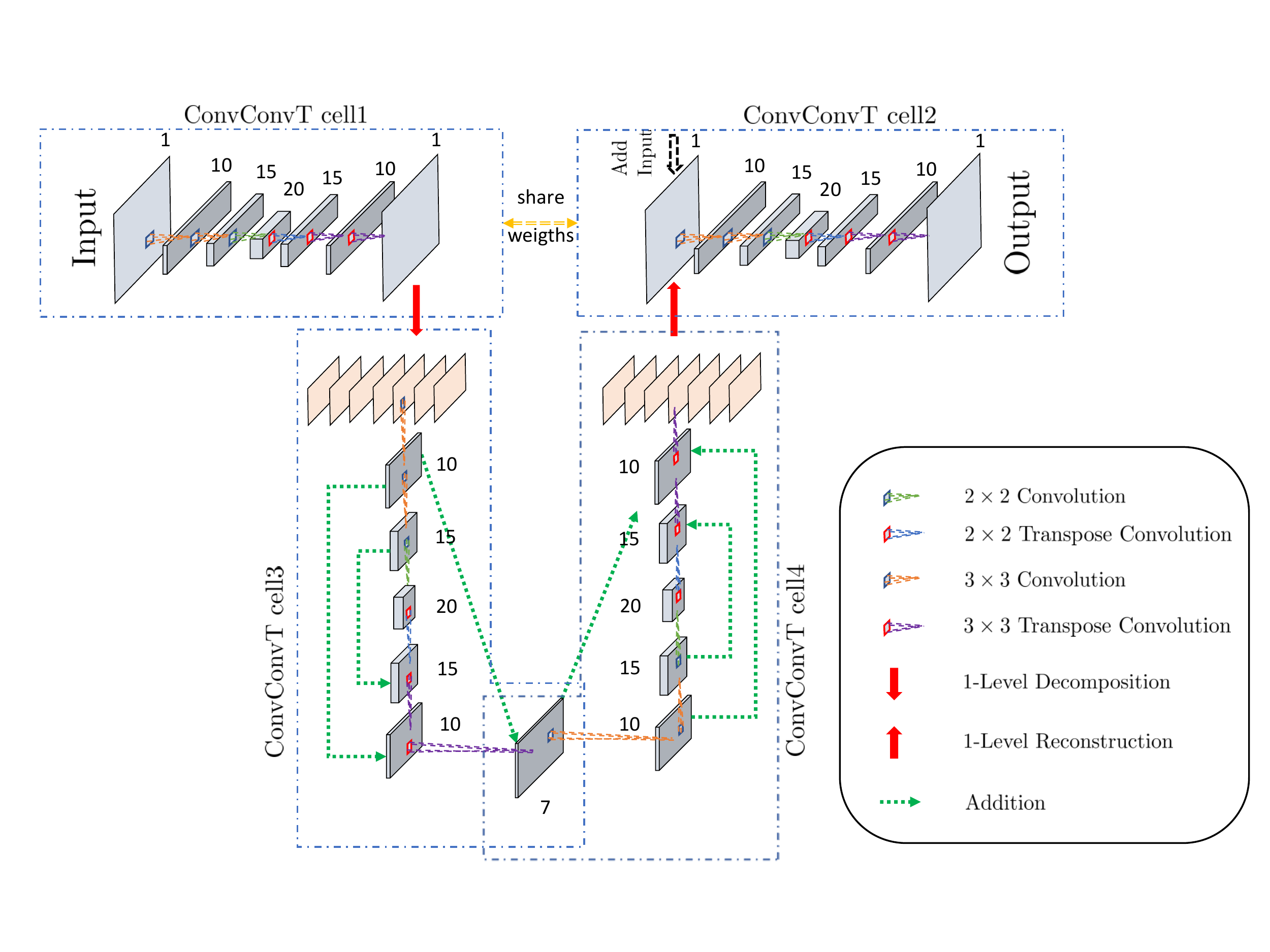}
	\caption{Neural Network Architecture. Numbers around each layer denote the channels. This neural network architecture can also be extended to deeper layers with more decompositions, reconstructions and ConvConvT cells cascaded. }\label{nn}
\end{figure*}

For the MNIST, CIFAR10 and Caltech101 experiment, we take three datasets: images from MNIST \cite{MNIST} with 60000 for training and  10000 for testing, images from CIFAR10 \cite{cifar10} with 50000 for training and 10000 for testing, and images from Caltech101 \cite{caltech101} with 7677 for training and 1000 for testing. During the training, we use ADAM algorithm and a mini-batch size $20$. Learning rate decays exponentially from the beginning value $0.005$ with multiplicative factor $0.9$ in $20$ epochs. Table \ref{mean_var} shows denoising mean and variance of PSNRs on the test datasets. It demonstrates that the proposed neural network significantly outperforms the threshold methods with high stability simultaneously. For instance, at the case of $rate=0.2$, the best PSNR for MNIST is $17.91$dB by soft threshold versus $25.21$dB by the CNN denoising. In Figure~\ref{mnist}, one can observe the mean values of PSNR on test datasets after each epoch over 20 independent trails.    %much better on spherical MNIST than denoising functions, and it is rather stable.
%\begin{table}[H]
%\centering
%\caption{Mean and Variance of denoising NN over 5 runs.}
%\begin{tabular}{|c|c|c|}
%\hline
%\textbf{Rate} & \textbf{Mean} & \textbf{Variance} \\ \hline
%\textbf{0.05} & 33.89594      & 0.21746           \\ \hline
%\textbf{0.1}  & 29.4911       & 0.08936           \\ \hline
%\textbf{0.2}  & 25.22524      & 0.00731           \\ \hline
%\textbf{0.3}  & 22.75501      & 0.00843           \\ \hline
%\textbf{0.5}  & 19.66276      & 0.02257           \\ \hline
%\end{tabular}
%\end{table}

\begin{table}[http]
	\centering
	%\footnotesize
	\scriptsize
	%\tiny
	\setlength\tabcolsep{2pt}
	\caption{Mean and Variance of denoising NN over 20 independent trials.}
	\label{mean_var}
	\begin{tabular}{|c|c|c|c|c|c|}
		\hline
		Dataset & Rate     & \multicolumn{1}{c|}{\textbf{0.05(26)}} & \multicolumn{1}{c|}{\textbf{0.1(20)}} & \multicolumn{1}{c|}{\textbf{0.2(14)}}  & \multicolumn{1}{c|}{\textbf{0.5(6)}} \\ \hline
		\multirow{2}{*}{\textbf{MNIST}}      & Mean     & 33.87(7.85$\uparrow$)                   & 29.43(9.43$\uparrow$)                  & 25.21(11.23$\uparrow$)                                    & 19.84(13.81$\uparrow$)                \\ \cline{2-6}
		& Variance & 0.46                          & 0.22                         & 0.02                                                    & 0.00                        \\ \hline
		\multirow{2}{*}{\textbf{CIFAR10}}    & Mean     & 28.48(2.46$\uparrow$)                   & 25.35(5.34$\uparrow$)                  & 22.17(8.19$\uparrow$)                                     & 18.49(12.46$\uparrow$)                \\ \cline{2-6}
		& Variance & 0.51                          & 0.07                         & 0.01                                                    & 0.00                        \\ \hline
		\multirow{2}{*}{\textbf{Caltech101}} & Mean     & 29.46(3.44$\uparrow$)                   & 26.42(6.42$\uparrow$)                  & 23.58(9.60$\uparrow$)                                     & 20.52(14.50$\uparrow$)                \\ \cline{2-6}
		& Variance & 0.29                          & 0.04                         & 0.02                                                  & 0.07                        \\ \hline
	\end{tabular}
\end{table}

\begin{figure}[http]
	\centering
	\subfigure[]{\includegraphics[width=2.5cm]{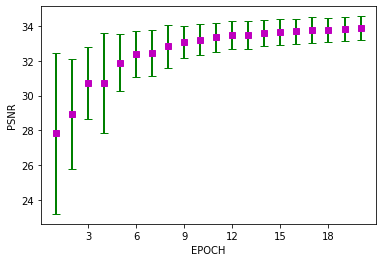}}
	%\subfigure[]{\includegraphics[width=2.5cm]{mn01.png}}
	\subfigure[]{\includegraphics[width=2.5cm]{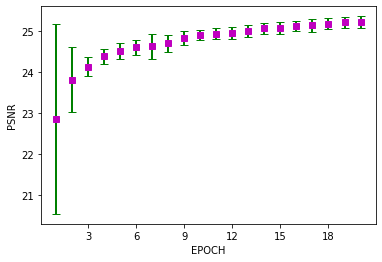}}
	%\subfigure[]{\includegraphics[width=2.5cm]{mn03.png}}
	\subfigure[]{\includegraphics[width=2.5cm]{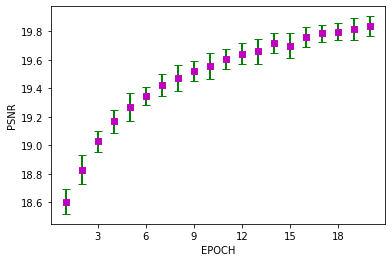}}
	\\
	\subfigure[]{\includegraphics[width=2.5cm]{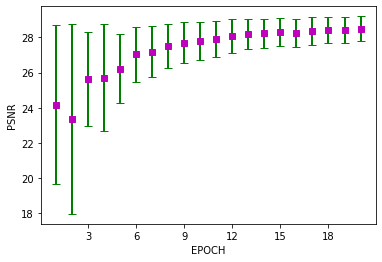}}
	%\subfigure[]{\includegraphics[width=2.5cm]{ci01.png}}
	\subfigure[]{\includegraphics[width=2.5cm]{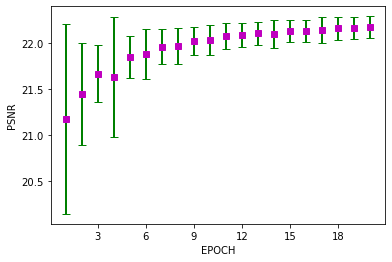}}
	%\subfigure[]{\includegraphics[width=2.5cm]{ci03.png}}
	\subfigure[]{\includegraphics[width=2.5cm]{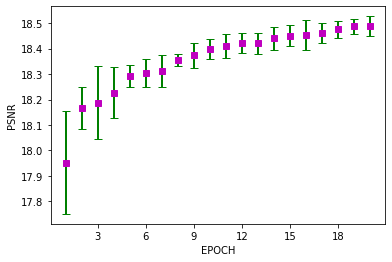}}
	\\
	\subfigure[]{\includegraphics[width=2.5cm]{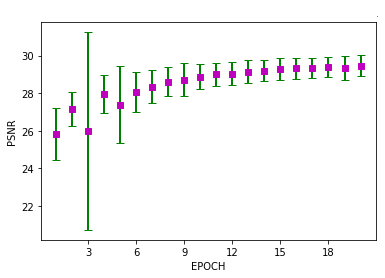}}
	%\subfigure[]{\includegraphics[width=2.5cm]{ca01.png}}
	\subfigure[]{\includegraphics[width=2.5cm]{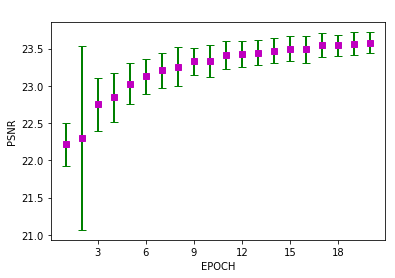}}
	%\subfigure[]{\includegraphics[width=2.5cm]{ca03.png}}
	\subfigure[]{\includegraphics[width=2.5cm]{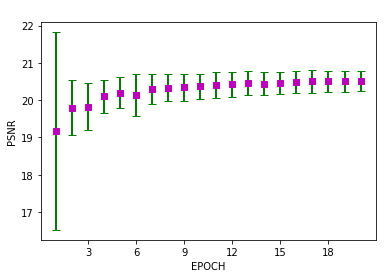}}
	\caption{Figure (a-c), (d-f) and (g-i) show PSNRs on test sets of MNIST, CIFAR10 and Caltech101 respectively with datasets having noise rate 0.05, 0.2 and 0.5 from left colum to right column after each epoch over $20$ independent trails. Errorbars at data markers are standard variance.}
	\label{mnist}
\end{figure}

With the trained model by the dataset Caltech101, we apply it on the classical gray scale images (``Barbara'', ``Boat'', ``Hill'', and ``Man'') that are in size $512\times 512$. Despite these spherical gray scale images are of great difference with our training datasets in size and texture, the trained network is still able to increase PSNRs and outperforms threshold methods. For instance, comparing Table~\ref{level4} and~\ref{psnrm} at the case $rate=0.1$, the best PSNR for the image ``Boat'' is $26.24$dB by bivariate threshold versus $28.61$dB by the CNN denoising.
To further evaluate the robustness of the proposed neural network, we carry out 20 independent training and calculate the mean and variance  in terms of PSNRs. Table~\ref{psnrm} illustrates the results, and demonstrates the generalization ability of the neural network. %Although spherical gray scale images are not used as training set, neural network improve the PSNR on these images.

\begin{table}[htp]
	\centering
	%\tiny
	\scriptsize
	\setlength\tabcolsep{2pt}
	%\linespread{1.5}
	\caption{Mean and Varience of PSNRs by denoising NN over $20$ different trials.}
	\begin{tabular}{|c|c|c|c|c|c|}
		\hline
		 Image &Rate    & \textbf{0.05(26)}    & \textbf{0.1(20) }    & \textbf{0.2(14)}          & \textbf{0.5(6)}       \\ \hline

		 \multirow{2}{*}{\textbf{Barbara}} & Mean & 30.11(4.09$\uparrow$) & 26.68(6.68$\uparrow$) & 23.75(9.77$\uparrow$)  & 21.22(15.20$\uparrow$) \\ \cline{2-6}
		 & Variance & 0.18     & 0.03    & 0.02          & 0.08 \\ \hline
		 \multirow{2}{*}{\textbf{Boat}} & Mean    & 31.33(5.30$\uparrow$) & 28.61(8.61$\uparrow$) & 26.04(12.06$\uparrow$) & 22.81(16.78$\uparrow$) \\ \cline{2-6}
		 & Variance & 0.18     & 0.05    & 0.03          & 0.13 \\ \hline
		 \multirow{2}{*}{\textbf{Hill}} & Mean   & 30.89(4.87$\uparrow$) & 28.31(8.32$\uparrow$) & 25.97(12.00$\uparrow$)  & 23.09(17.07$\uparrow$) \\ \cline{2-6}
		 & Variance & 0.12     & 0.04    & 0.02          & 0.14  \\ \hline
		 \multirow{2}{*}{\textbf{Man }} & Mean   & 31.55(5.53$\uparrow$) & 28.80(8.80$\uparrow$) & 26.31(12.33$\uparrow$)  & 23.16(17.14$\uparrow$) \\ \cline{2-6}
		 & Variance &  0.21     & 0.05    & 0.02          & 0.16 \\ \hline
	\end{tabular}
	%\caption{average PSNRs of denoising NN over $20$ different trial.}
	\label{psnrm}
\end{table}

\section{Conclusion and further remarks}
\label{sec:remarks}
In this paper, we present a general framework for the construction of Haar tight framelet on any compact set associated with a hierarchical partition.
Specifically, an area-regular hierarchical partition on the 2-sphere is constructed and the corresponding Haar frames are defined consequently. We conduct denoising experiments on dataset-ETOPO, spherical gray images, and spherical MNIST, CIFAR10 and Caltech101 by employing thresholding methods and CNN. The experiment results show that techniques based on our Haar tight framelets can provide competitive performance on denoising by comparing with existing results. Particularly, the proposed CNN significantly outperforms the existing threshold methods.% and CNN models, such as MWCNN \cite{CNN_Liu}.
%A comparison with MWCNN in \cite{Liu} is made for which using 20 thousands parameters

We remark that (a) our framework is general and the Haar tight framelets can be constructed on any compact sets including intervals, squares, cubes, and compact manifolds; (b) The orthonormal basis in our main results is  Haar-type (characteristic functions) but it could extend to non-Haar types by considering more general orthonormal bases on the concerned domains; (c) We propose a novel partition on the sphere here and we point out that other partition methods on the sphere can be used in view of our general framework; (d) The data we process are on the whole sphere but we point out that our method could pin-point to process local (spherical) data only; (e) our work can be efficiently adapted to   real-world datasets collected by Omnidirectional camera, Fermi gamma ray space telescope, autonomous vehicles, and so on, which we will consider in our future work.

\section{Appendix}
\label{sec:appendix}

We provide proofs of Lemma~\ref{lem2.1}, Theorem~\ref{Thm2.2}, and Corollary~\ref{cor:area-regular}.

\begin{proof}[Proof of Lemma~\ref{lem2.1}]
	For the sufficiency of (i), note that if $\Psi$ is a tight frame  for $\mathcal{W} := \Span\Psi $ with frame bound $c$ then
%	\begin{equation*}\label{tight_frame}
		$f = \f{1}{c} \sum_{i = 1}^{n}  \langle f, \psi_i \rangle \psi_i$ for all $ f \in \mathcal{W}$.
%	\end{equation*}
	For any $k\in[n]$, taking $f=\psi_k$ and  by the orthogonality of $\{\phi_1,\ldots, \phi_\ell\}$, we have that
	\begin{align*}
		\psi_k
		&=\f{1}{c} \sum_{i = 1}^{n}  \langle \psi_k, \psi_i \rangle \psi_i \\
		&=  \f{1}{c} \sum_{i=1}^{n}  \left\langle \sum_{p=1}^\ell  a_{k,p} \phi_p, \sum_{q = 1}^\ell a_{i,q}\phi_q \right\rangle \sum_{m = 1}^\ell a_{i,m}\phi_m  \notag \\
		&=  \f{1}{c}  \sum_{m=1}^\ell \phi_m \sum_{i=1}^n \sum_{p = 1}^\ell a_{k,p} a_{i,p} a_{i,m}.
	\end{align*}
	On the other hand,  since $ \psi_k = \sum_{m=1}^\ell a_{km} \phi_m $, it yields that 	$a_{k,m} = \f{1}{c} \sum_{i=1}^{n} \sum_{p=1}^{\ell} a_{k,p} a_{i,p} a_{i,m}$ for any $k\in[n]$ and $m\in[\ell]$,
	and thus $\b A=\frac 1 c \bm{A} \bm{A}^\top \bm{A} $. Using the same argument, the necessity can be proved as well.

	For (ii), by the orthogonality of $\{\phi_1,\ldots, \phi_\ell\}$, the statement that $\phi_0\in \mathcal W^{\bot}$ is equivalent to that for all $k\in[n]$,
	\begin{align*}
		\langle \psi_k, \phi_0 \rangle
		&= \left\langle \sum\limits_{j=1}^\ell a_{k,j} \phi_j, \phi_0 \right\rangle  =\sum\limits_{j=1}^\ell a_{k,j}p_j,
	\end{align*}
	which means $\b{A}\b{p}=\b{0}$. On the other hand, it is trivial that  $\Span\{ \phi_1, \dots, \phi_\ell \} \supseteq  \Span \{ \phi_0, \psi_1, \dots, \psi_n\}$. The converse is equivalent to show that there exists $\bm{Q}\in \RR^{\ell\times(n+1)}$ such that
	\[\begin{bmatrix} \phi_1 \\ \vdots \\ \phi_\ell \end{bmatrix}=\b{Q}\begin{bmatrix} \phi_0 \\  \psi_1  \\ \vdots \\ \psi_n \end{bmatrix}=\b{Q}\left[
	\begin{array}{c}
		\b{p}^\top \\
		\b{A} \\
	\end{array}
	\right]\begin{bmatrix} \phi_1 \\ \vdots \\ \phi_\ell \end{bmatrix},\]
	which implies that $\b{Q}\left[
	\begin{array}{c}
		\b{p}^\top \\
		\b{A} \\
	\end{array}
	\right]=\b{I}_\ell$.
\end{proof}

%Below is the proof of Theorem~\ref{Thm2.2}.
\begin{proof}[Proof of Theorem~\ref{Thm2.2}]
	For $j\in\NN_0$, let $\mathcal{V}_{j} := \Span\{ \phi_{\vec v} \setsep \vec v\in \Lambda_j \}$ and $\mathcal{W}_{j} := \Span\{ \psi_{(\vec v,k)} \setsep \vec v\in\Lambda_j, k\in [n_j] \}$. Noting that for each $\vec v\in \Lambda_j$, the system $\{\phi_{(\vec v,i)}\}_{i\in [\ell_{j+1}]}$ is orthonormal.  By Lemma~\ref{lem2.1} we have that
	\[
	\Span\{\phi_{(\vec v,i)}: i\in[\ell_{j+1}]\}=\Span\{\phi_{\vec v}\}\oplus \Span\{\psi_{(\vec v, k)}: k\in[n_j]\},
	\]
	which yields that
	\[
	\CV_{j+1}=\CV_j\oplus \CW_{j}.
	\]
	Iteratively, for any $0\leq L< J $,
	\[
	\CV_J=\CV_L\oplus \CW_L\oplus \CW_{L+1}\oplus\cdots\oplus \CW_{J-1}.
	\]
	In addition, by the nested property and the density property of the hierarchical partition, we have $\mathcal{V}_0 \subseteq \mathcal{V}_1 \subseteq \dots \subseteq \mathcal{V}_j \subseteq \cdots$, and $ \cup_{j = 0}^{\infty} \mathcal{V}_j $ is dense in $ L_2(\Omega)$, which imply that
	$\CF_L(\{\CB_j\}_{j\in\NN})$ is a tight frame for $L_2(\Omega)$.
\end{proof}
%Below is the proof of Corollary~\ref{cor:area-regular}.
\begin{proof}[Proof of Corollary~\ref{cor:area-regular}]
	Let $\b P_i$, $i=1,\ldots,n$ be the permutation matrices such that all $\b P_i \b w $ are distinct. Then $\b A=[\b P_1 \b w, \dots, \b P_n \b w]^\top$ and $\b A^\top \b A = \sum_{i=1}^n \b P_i \b w \b w^\top \b P^\top_i$.
	We claim that $\b A^\top \b A$ is in the following form:
	\[\begin{bmatrix}
	x & y & \dots & y \\
	y & x & \ddots &\vdots \\
	\vdots & \ddots & \ddots & y \\
	y & \dots & y & x
	\end{bmatrix},\]
	for some $x,y$. With the claim, for each $i=1,\ldots, n$,
	\begin{align*}
		\b w^\top \b P_i^\top \b A^\top \b A=&\b w^\top \b P_i (x-y)\b I+y\b w^\top \mathds{1}_{\ell\times \ell}\\
		= &(x-y)\b w^T \b P_i,
	\end{align*}
	where $\mathds{1}_{\ell\times \ell}$ is the matrix of all entries $1$ and the second equality follows from $\b w^\top \b 1_{\ell}=0$. It yields that $\b A\b A^\top\b A=(x-y) \b A$. To see the existence of $\b Q$, one can just notice the row rank of $\binom{\b 1_\ell^\top}{\b A}$ is full since the row rank of $\b A$ is $\ell-1$ and $\b 1_{\ell}\perp row(\b A)$.
	To the end, it only remains to show the above claim. For all $i=1,\ldots,n$, let $\s_i$ be the corresponding permutation of $\b P_i$ on $[\ell]$, and let $\b B:=\b w^\top \b w$. Note that for any $1\leq s, t\leq n$, $(\b P_i \b B \b P_i^\top)_{(s,t)}=B_{(\s_i s, \s_i t)}$. If we denote $\b M:=\b A^\top \b A$, this implies that
$\b M_{s,t}=\sum_{i=1}^n B_{(\s_i s, \s_i t)}$ for all $1\leq s,t\leq n$. Such summation of entries of $\b B$ under permutation can guarantee the diagonal entries of $\b M$ to be same. For any $1\leq s\neq t\leq n$, with the similar argument the value of $\b M_{s,t}$ is independent of the choice of $s,t$. It leads to the claim.
\end{proof}

\section*{Acknowledgments}
The authors thank the anonymous reviewers for their constructive comments and valuable suggestions that greatly help the improvement of the quality of the paper. The research and the
work described in this paper was partially supported by grants from the Research Grants Council of the Hong Kong Special Administrative Region, China [Projects Nos. CityU 11306220, CityU 11302218, and C1013-21GF] and City University of
Hong Kong [Project Nos. 7005497 and 7005603].

\bibliographystyle{IEEEtran}

\end{document}